\def\draft{0}
\def\llncs{0}

\ifnum\llncs=1

\else

\fi

\documentclass[11pt]{article}
\usepackage[margin=0.7in]{geometry}
\usepackage{amssymb,amsmath,cite,url,enumerate,algorithm,algorithmic,latexsym,amsfonts,amsthm}
\usepackage[table]{xcolor}

\newtheorem{theorem}{Theorem}[section]

\usepackage[normalem]{ulem}

\ifnum\draft=1
\def\ShowAuthNotes{1}
\else
\def\ShowAuthNotes{0}
\fi

\ifnum\ShowAuthNotes=1
\newcommand{\authnote}[3]{\textcolor{#3}{[{\footnotesize {\bf #1:} { {#2}}}]}}
\else
\newcommand{\authnote}[3]{}
\fi

\definecolor{DarkBlue}{RGB}{0,0,150}
\definecolor{DarkRed}{RGB}{150,0,0}
\definecolor{DarkGreen}{RGB}{0,150,0}

\usepackage{amsthm,amssymb} %
\usepackage{fullpage}
\usepackage{booktabs} %
\usepackage{dsfont} %
\usepackage{float} %
\usepackage{subcaption}

\usepackage{xcolor}
\definecolor{darkgray}{rgb}{0.25, 0.25, 0.25}

\usepackage{microtype} %
\usepackage{charter} %
\usepackage{soul} %
\usepackage{suffix} %

\usepackage{xcolor} %

\usepackage{amsmath}

\usepackage{multicol}
\usepackage{stackengine}
\usepackage{scalerel}

\newtheorem{lemma}[theorem]{Lemma} %
\newtheorem{proposition}[theorem]{Proposition} %
\newtheorem{corollary}[theorem]{Corollary} %
\newtheorem{definition}[theorem]{Definition} %
\newtheorem{claim}[theorem]{Claim} %

\usepackage[utf8]{inputenc} %
\usepackage{mdframed} %
\usepackage{graphicx} %
\usepackage{hyperref} %
\usepackage{cleveref}
\usepackage[inline]{enumitem}%

\usepackage{fancybox}

\renewcommand{\hat}[1]{\widehat{#1}} %
\renewcommand{\tilde}[1]{\widetilde{#1}} %

\newcommand{\calS}{\mathcal{S}}

\newcommand{\complex}{\mathbb{C}} %
\let\epsilon=\varepsilon %
\newcommand{\eps}{\epsilon} %

\newcommand{\microspace}{\mspace{.5mu}} %
\newcommand{\ket}[1]{\ensuremath{\lvert\microspace #1
    \microspace\rangle}} %
\newcommand{\bra}[1]{\ensuremath{\langle\microspace #1
    \microspace\rvert}} %
\newcommand{\ketbra}[2]{\lvert #1 \rangle \! \langle #2 \rvert} %
\newcommand{\braket}[2]{\langle #1 |#2 \rangle} %
\newcommand{\kb}[1]{\ketbra{#1}{#1}} %
\newcommand{\set}[1]{\{#1\}}

\newcommand{\comment}[1]{}

\usepackage{xspace} %
\newcommand{\class}[1]{\textup{#1}\xspace} %

\newcommand{\NP}{\class{NP}} %
\newtheorem{fact}[theorem]{Fact} %

\definecolor{protocol-bg}{RGB}{240,240,240}

\newcommand{\setft}[1]{\mathrm{#1}} %

\DeclareMathOperator{\poly}{poly}
\DeclareMathOperator{\negl}{negl}

\hypersetup{pagebackref, colorlinks=true, urlcolor=blue,
  linkcolor=blue, citecolor=DarkRed}
\usepackage{tikz} %
\usetikzlibrary{calc,positioning}
\usetikzlibrary{decorations.pathreplacing}

\mdfdefinestyle{figstyle}{ %
  linecolor=black!7, %
  backgroundcolor=black!7, %
  innertopmargin=10pt, %
  innerleftmargin=25pt, %
  innerrightmargin=25pt, %
  innerbottommargin=10pt %
}

\definecolor{White}{rgb}{1,1,1} %
\definecolor{Black}{rgb}{0,0,0} %
\definecolor{LightGray}{rgb}{.8,.8,.8} %
\colorlet{ChannelColor}{LightGray} %
\colorlet{ChannelTextColor}{Black} %
\colorlet{ReadoutColor}{White} %

\newcommand{\cS}{\mathcal S}

\def\01{\{0,1\}}

\usepackage{soul}

\newcommand{\func}{\mathcal{F}}
\newcommand{\smallset}[1]{\{#1\}}
\renewcommand{\set}[1]{\left\{#1\right\}}

\newcommand{\ecom}{\mathsf{ECom}}
\newcommand{\icom}{\langle \comc, \comr \rangle}
\newcommand{\comc}{\mathcal{C}}
\newcommand{\comr}{\mathcal{R}}

\newcommand{\bits}{\{0,1\}}
\newcommand{\zo}{\{0,1\}}

\newcommand{\party}{\mathcal{P}}

\newcommand{\funcG}{\mathcal{G}} %
\newcommand{\adversary}[1]{\mathcal{#1}} %
\newcommand{\adv}{\adversary{A}} %
\newcommand{\simulator}{\adversary{S}} %

\newcommand{\envt}{\adversary{Z}} %
\newcommand{\secpar}{\lambda} %
\newcommand{\mac}{{M}} %
\newcommand{\qreg}[1]{\mathsf{#1}} %
\newcommand{\rspace}[1]{\mathsf{#1}} %
\newcommand{\td}{\text{TD}}
\newcommand{\veps}{\varepsilon}
\newcommand{\identity}{\mathbb{I}}
\newcommand{\cprot}[3]{ {#1}^{{#2}/{#3}}} %
\newcommand{\lin}[1]{\setft{L}\left(#1\right)}
\newcommand{\density}[1]{\setft{D}\left(#1\right)}

\newcommand{\qc}{\approx_{qc}} %
\newcommand{\dmm}{\approx_{\diamond}} %
\newcommand{\cqsa}{\text{\tt C-QSA}}

\newcommand{\sqsa}{\text{\tt S-QSA}}

\newcommand{\HH}{\ensuremath{\mathbf{H}}}
\newcommand{\bd}[1]{{#1}} %
\newcommand{\bk}[1]{|#1\rangle \langle #1|} %

\newcommand{\funct}[1]{\mathcal{F}_{#1}}

\newcommand{\socom}{\text{\tt so-com}}
\newcommand{\fsocom}{\funct{\socom}}

\newcommand{\com}{\text{\tt com}}
\newcommand{\fcom}{\funct{\com}}

\newcommand{\pot}{\text{\tt p-ot}}
\newcommand{\fpot}{\funct{\pot}}

\newcommand{\cds}{\mathsf{cds}}
\newcommand{\fcds}{\funct{\cds}}

\newcommand{\ot}{\text{\tt ot}}
\newcommand{\fot}{\funct{\ot}}

\newcommand{\zk}{\text{\tt zk}}
\newcommand{\fzk}{\funct{\zk}}

\newcommand{\owf}{OWF}
\newcommand{\pqOWF}{pqOWF}

\newcommand{\base}{\{+,\times\}}
\newcommand{\qotepr}{\Pi_{\tt QOT}^{\tt EPR}}
\newcommand{\qotbbcs}{\Pi_{\tt QOT}}
\newcommand{\qothzero}{\Pi^{\tt H_0}_{\tt QOT}}
\newcommand{\qothone}{\Pi^{\tt H_1}_{\tt QOT}}
\newcommand{\qothtwo}{\Pi^{\tt H_2}_{\tt QOT}}

\newcommand{\res}[2]{{#1}|_{#2}} %

\newcommand{\qsocomhzero}{\Pi^{\tt H_0}_{\socom}}
\newcommand{\qsocomhone}{\Pi^{\tt H_1}_{\socom}}

\newcommand{\Lang}{\mathcal{L}}
\newcommand{\Rel}{\mathcal{R}}
\newcommand{\Ver}{\mathsf{Ver}}
\newcommand{\Nat}{\mathbb{N}}
\newcommand{\inter}[1]{\langle #1 \rangle}
\newcommand{\Garb}{\mathsf{Garb}}
\newcommand{\GEval}{\mathsf{GEval}}

\newcommand{\err}{\mathsf{err}}

\newcommand{\cdsS}{\mathsf{cds.S}}
\newcommand{\cdsR}{\mathsf{cds.R}}

\newcommand{\qecomhzero}{\Pi_{\tt H_0}^{\ecom}}
\newcommand{\qecomhone}{\Pi_{\tt H_1}^{\ecom}}
\newcommand{\qecomhtwo}{\Pi_{\tt H_2}^{\ecom}}
\newcommand{\qecomhthree}{\Pi_{\tt H_3}^{\ecom}}

\newcommand{\otS}{\mathsf{ot.S}} \newcommand{\otR}{\mathsf{ot.R}}
\newcommand{\comS}{\mathsf{com.S}} \newcommand{\comR}{\mathsf{com.R}}
\newcommand{\otSimR}{\mathsf{ot.SimR}}
\newcommand{\otSimS}{\mathsf{ot.SimS}}

\newcommand{\comSimS}{\mathsf{com.SimS}}
\renewcommand{\paragraph}[1]{\smallskip\noindent{\bf #1}}
\renewcommand{\subparagraph}[1]{\smallskip\noindent\underline{\sc #1}}
\newcommand{\Note}[1]{{\it \small \textsc{Note:} #1}}
\newcommand{\highlight}[1]{\underline{\em #1}}

\newcommand{\fsnote}[1]{\authnote{Fang}{#1}{cyan}}
\newcommand{\agnote}[1]{\authnote{Alex}{#1}{DarkBlue}}
\newcommand{\vnote}[1]{\authnote{Vinod}{#1}{magenta}}

\newcommand{\rachel}[1]{\authnote{Rachel}{#1}{DarkGreen}}

\def\mute{1}

\ifnum\mute=0
\newcommand{\authnotem}[3]{\textcolor{#3}{[{\footnotesize {\bf #1:} { {#2}}}]}}
\else
\newcommand{\authnotem}[3]{}
\fi

\newcommand{\fsnotem}[1]{\authnotem{Fang}{#1}{cyan}}
\newcommand{\agnotem}[1]{\authnotem{Alex}{#1}{DarkBlue}}
\newcommand{\vnotem}[1]{\authnotem{Vinod}{#1}{magenta}}
\newcommand{\rnotem}[1]{\authnotem{Rachel}{#1}{DarkGreen}}
\newcommand{\rachelm}[1]{\authnotem{Rachel}{#1}{DarkGreen}}

\renewcommand{\subparagraph}[1]{\smallskip\noindent\underline{\textsc{#1}}}

\title{Oblivious Transfer is in MiniQCrypt}
\author{Alex B. Grilo\thanks{Sorbonne Universit\'{e}, CNRS, LIP6
\quad $^\dagger$ University of Washington  \quad $^\ddag$ Portland State University \quad $^\P$ MIT 
} \and Huijia Lin$^\dagger$ \and Fang Song$^\ddag$ \and Vinod Vaikuntanathan$^\P$}

\begin{document}

\ifnum\draft=1

\paragraph{TODO. Update: 10/07}
\newcommand{\leader}[1]{(\textbf{#1})}
\begin{itemize}
\item \leader{?} Intro
\begin{itemize}
    \item technical review: we need to explain the novelty on the ecom (Rachel?) -- try to explain the subtlety why IPS08 does not work (from email)
    \item related works (Alex): maybe discuss IPS here and why it does not solve it; 
    \fsnote{finished the following.} \sout{other quantum protocols that were not discussed; OT in other models; lower-bounds on quantum protocols; other post-quantum secure protocols (and their assumptions)}
\end{itemize}
\item Prelim: move quantum to appendix, finish garbled circuits, polish everything (Alex)

    \item \leader{all} proofread 3.3. 

    \item \leader{Vinod} section 4.
    \item \leader{Fang} clean up bib, some duplicates \agnote{Did it} 
\end{itemize}
 \fi

\maketitle

\vspace{-1.5em}
\begin{abstract}
 MiniQCrypt is a world where quantum-secure one-way functions exist, and quantum communication is possible. We construct an oblivious transfer (OT) protocol in MiniQCrypt that achieves simulation-security in the plain model against malicious quantum polynomial-time adversaries, building on the foundational work of Bennett, Brassard, Cr\'{e}peau and Skubiszewska (CRYPTO 1991). Combining the OT protocol with prior works, we obtain secure two-party and multi-party computation protocols also in MiniQCrypt. This is in contrast to the classical world, where it is widely believed that one-way functions alone do not give us OT. 
 
 In the common random string model, we  achieve a {\em constant-round} universally composable (UC) OT protocol. \vnote{Are the 2PC and MPC protocols also constant-round in CRS?} \rachel{given UC OT, we will have constant-round post-quantum MPC. Does that give constant round quantum MPC?}\vnote{that's a good start! what's the reference btw? we should claim it.}\rachel{We can just cite IPS.}
\end{abstract}
 
\thispagestyle{empty}
\newpage
\tableofcontents
\pagenumbering{roman}
\newpage
\pagenumbering{arabic}

\newpage

\vspace{-2em}
\section{Introduction}
\vspace{-0.5em}

Quantum computing and modern cryptography have enjoyed a highly
productive relationship for many decades ever since the conception of
both fields. On the one hand, (large-scale) quantum computers can be
used to break many widely used cryptosystems based on the hardness of
factoring and discrete logarithms, thanks to Shor's
algorithm~\cite{Shor94}. On the other hand, quantum information and
computation have helped us realize cryptographic tasks that are
otherwise impossible, for example quantum money~\cite{Wiesner} and generating certifiable randomness~\cite{Colbeck09,VV12,BCMVV18}.

Yet another crown jewel in quantum cryptography is the discovery, by Bennett and Brassard~\cite{BB84}, of a key exchange protocol whose security is unconditional. That is, they achieve information-theoretic security for a cryptographic task that classically necessarily has to rely on unproven computational assumptions. In a nutshell, they accomplish this using the uncloneability of quantum states, a bedrock principle of quantum mechanics. What's even more remarkable is the fact that their protocol makes minimalistic use of quantum resources, and consequently, has been implemented in practice over very large distances~\cite{Dixon_2008,Liao_2018}. This should be seen in contrast to large scale quantum {\em computation} whose possibility is still being actively debated.

Bennett and Brassard's groundbreaking work raised a {\em tantalizing} possibility for the field of cryptography:
\begin{quote}
\centering
  {\em Could {\em every} cryptographic primitive\\ be realized unconditionally using quantum information?}
\end{quote}

A natural next target is oblivious transfer (OT), a versatile cryptographic primitive which, curiously, had its origins in Wiesner's work in the 1970s on quantum information~\cite{Wiesner} before being rediscovered in cryptography by Rabin~\cite{Rabin81} in the 1980s. Oblivious transfer (more specifically, $1$-out-of-$2$ OT) is a two-party functionality where a receiver Bob wishes to obtain one out of two bits that the sender Alice owns. The OT protocol must ensure that Alice does not learn which of the two bits Bob received, and that Bob learns only one of Alice's bits and no information about the other.
Oblivious transfer lies at the foundation of secure computation, allowing us to construct protocols for the secure multiparty computation (MPC) of any polynomial-time computable function~\cite{STOC:GolMicWig87,STOC:Kilian88,C:IshPraSah08}.

Bennett, Brassard, Cr\'{e}peau and Skubiszewska\cite{C:BBCS91} constructed an OT protocol given an {\em ideal} bit commitment protocol and quantum communication. In fact, the only quantum communication in their protocol consisted of Alice sending several so-called ``BB84 states'' to Bob. Unfortunately, {\em unconditionally secure} commitment~\cite{Mayers97,LoChau97} and {\em unconditionally secure} OT~\cite{Lo97,CGS16} were soon shown to be impossible even with quantum resources.

However, given that bit commitment can be constructed from one-way functions (\owf{})~\cite{C:Naor89,HILL99}, the hope remains that OT, and therefore a large swathe of cryptography, can be based on only {\em \owf{}} together with (practically feasible) quantum communication. Drawing our inspiration from Impagliazzo's five worlds in cryptography~\cite{Imp95}, we call such a world, where post-quantum secure one-way functions (\pqOWF{}) exist and quantum computation and communication are possible,  {Mini\textbf{\textcolor{blue}{Q}}Crypt}. The  question that motivates this paper is:
\begin{quote}
  \begin{center}
  {\em Do OT and MPC exist in MiniQCrypt?}
  \end{center}
\end{quote}

Without the quantum power, this is widely believed to be
impossible. That is, given only \owf{}s, there are no {\em
  black-box} constructions of OT or even key exchange
protocols~\cite{IR89,Rud91}. The fact that \cite{BB84} overcome this
barrier and construct a key exchange protocol with quantum
communication (even without the help of \owf{}s)
reinvigorates our hope to do the same for OT.

\paragraph{Aren't We Done Already?}
At this point, the reader may wonder why we do not have an affirmative
answer to this question already, by combining the OT protocol of \cite{C:BBCS91}
based on bit commitments, with a construction of bit commitments from
\pqOWF{}~\cite{C:Naor89,HILL99}. Although this possibility was
mentioned already in \cite{C:BBCS91}, where they note that
``\ldots computational complexity based quantum cryptography is interesting
since it allows to build oblivious transfer around one-way
functions.'', attaining this goal remains elusive as we explain below.

First, proving the security of the \cite{C:BBCS91} OT protocol (regardless of the
assumptions) turns out to be a marathon.
After early proofs against {limited}
adversaries~\cite{MaySal94,STOC:Yao95}, it is relatively recently that we have a
clear picture with formal
proofs against arbitrary quantum polynomial-time adversaries~\cite{DFR+07,DFL+09,C:BouFeh10,EC:Unruh10}.
Based on these results, we
can summarize the state of the art as follows.

\begin{itemize}
\item \highlight{Using Ideal Commitments:}
  If we assume an \emph{ideal} commitment protocol, formalized as
  universally composable (UC) commitment, then the quantum OT
  protocol can be proven secure in strong
  {simulation}-based models, in particular the quantum UC model
  that admits sequential
  composition or even concurrent composition in a network setting~\cite{DFL+09,FS09,C:BouFeh10,EC:Unruh10}.
  However, UC commitments, in contrast to
  vanilla computationally-hiding and statistically-binding commitments,
  are powerful objects that do not live in Minicrypt. In particular,
  UC commitments give us key exchange protocols and are therefore
  black-box separated from Minicrypt.\footnote{The
  key exchange protocol between Alice and Bob works as follows. Bob, playing
  the simulator for a malicious
  sender in the UC commitment protocol, chooses a common reference string (CRS)
  with a trapdoor $TD$ and sends
  the CRS to Alice. Alice, playing the sender in the commitment scheme,
  chooses a random $K$ and runs the
  committer algorithm. Bob runs the straight-line
  simulator-extractor (guaranteed by UC simulation)
  using the $TD$ to get $K$, thus ensuring that Alice and Bob have a
  common key. An eavesdropper Eve should not
  learn $K$ since the above simulated execution is indistinguishable from
  an honest execution, where $K$ is hidden.}

\item \highlight{Using Vanilla Commitments:}
    If in the \cite{C:BBCS91} quantum OT protocol we use a {\em vanilla}
    statistically-binding and computationally hiding commitment scheme, which
    exists assuming a \pqOWF{}, the
    existing proofs, for example \cite{C:BouFeh10}, fall short in two respects.

    First, for a malicious receiver, the proof of \cite{C:BouFeh10}
    constructs only an {\em in}efficient simulator. Roughly speaking,
    this is because the OT receiver in \cite{C:BBCS91} acts as a
    committer, and vanilla commitments are not extractable. Hence, we
    need an inefficient simulator to extract the committed value by
    brute force. Inefficient simulation makes it hard, if not
    impossible, to use the OT protocol to build other protocols (even
    if we are willing to let the resulting protocol have inefficient
    simulation). Our work will focus on achieving the standard
    ideal/real notion of security~\cite{Golbook} with efficient
    simulators.

    Secondly, it is unclear how to construct a simulator (even
    ignoring efficiency) for a malicious sender. Roughly speaking, the
    issue is that simulation seems to require that the commitment
    scheme used in \cite{C:BBCS91} be secure against selective opening
    attacks, which vanilla commitments do not
    guarantee~\cite{EC:BelHofYil09}.

  \item \highlight{Using Extractable Commitments:} It turns out that the
    first difficulty above can be addressed if we assume a commitment
    protocol that allows \emph{efficient extraction} of the committed
    value -- called extractable commitments.
    Constructing extractable commitments is surprisingly challenging
    in the quantum world because of the hardness of
    rewinding. Moreover, to plug into the quantum OT protocol, we need
    a strong version of extractable commitments from which the
    committed values can be extracted efficiently {\em without
      destroying or even disturbing the quantum states of the
      malicious committer,\footnote{This is because when using
        extractable commitment in a bigger protocol, the proof needs
        to extract the committed value and continue the execution with
        the adversary.}} a property that is at odds with quantum
    unclonability and rules out several extraction techniques used for
    achieving arguments of knowledge such as
    in~\cite{EC:Unruh12}.\rachelm{add citation to quantum AOK.} In
    particular, we are not aware of a construction of such extractable
    commitments without resorting to strong assumptions such as
    LWE~\cite{STOC:BitShm20,AnaPla19}, which takes us out of
    minicrypt.  Another standard way to construct extractable
    commitments is using public-key encryption in the CRS model, which
    unfortunately again takes us out of minicrypt.
\end{itemize}

To summarize, we would like to stress that before our work, the claims that quantum OT protocols can be constructed from \pqOWF{}s~\cite{C:BBCS91,FangUWYZ20} were rooted in misconceptions. \agnotem{Gave another shot}\vnotem{Controversial sentence?}

\paragraph{Why MiniQCrypt.}
Minicrypt is one of five Impagliazzo's worlds~\cite{Imp95} where \owf{}s exist,
but public-key encryption schemes do not. In Cryptomania, on the other hand, public-key
encryption schemes do exist.

Minicrypt is robust {\em and} efficient. It is robust because there is
an abundance of candidates for \owf{}s that draw from a
variety of sources of hardness, and most do not fall to quantum
attacks. Two examples are (\owf{}s that can be constructed
from) the advanced encryption standard (AES) and the secure hash
standard (SHA). They are ``structureless'' and hence typically do not
have any subexponential attacks either. In contrast, cryptomania seems
fragile and, to some skeptics, even endangered due to the abundance of
subexponential and quantum attacks, except for a handful of
candidates. It is efficient because the operations are combinatorial
in nature and amenable to very fast implementations; and the key
lengths are relatively small owing to \owf{}s against which
the best known attacks are essentially brute-force key search.  We
refer the reader to a survey by Barak~\cite{baraksurvey} for a deeper
perspective.

Consequently, much research in (applied) cryptography has been devoted
to minimizing the use of public-key primitives in advanced
cryptographic protocols~\cite{Beaver96,IKNP03}.  However, complete
elimination seems hard.  In the classical world, in the absence of
quantum communication, we can construct pseudorandom generators and
digital signatures in Minicrypt, but not key exchange, public-key
encryption, oblivious transfer or secure computation protocols. With
quantum {\em communication} becoming a reality not just
academically~\cite{Dixon_2008,Hiskett_2006,Pugh_2017} but also
commercially~\cite{Liao_2018}, we have the ability to reap the
benefits of robustness and efficiency that Minicrypt affords us, {\em
  and} construct powerful primitives such as oblivious transfer and
secure computation that were so far out of reach.

\paragraph{Our results.}
In this paper, we finally show that the longstanding (but previously unproved) claim is true.
\begin{theorem}[Informal]
  Oblivious transfer protocols in the plain model that are simulation-secure
  against malicious quantum polynomial-time adversaries
  exist assuming that post-quantum one-way functions exist and that quantum communication is possible.
\end{theorem}
Our main technical contribution consists of showing a construction of
an extractable commitment scheme based solely on \pqOWF{}s and using quantum communication. Our construction involves
three ingredients.  The first is vanilla post-quantum commitment
schemes which exist assuming that \pqOWF{}s
exist~\cite{C:Naor89}.  The second is post-quantum zero-knowledge
protocols which also exist assuming that \pqOWF{}s exist~\cite{Watrous09}.  The third and final ingredient is a
special multiparty computation protocol called conditional disclosure
of secrets (CDS) constructing which in turns requires OT.  This might
seem circular as this whole effort was to construct an OT protocol to
begin with!  Our key observation is that the CDS protocol is only
required to have a mild type of security, namely {\em unbounded
  simulation}, which {\em can} be achieved with a slight variant of
the \cite{C:BBCS91} protocol.  Numerous difficulties arise in our
construction, and in particular proving consistency of a protocol
execution involving quantum communication appears difficult: how do we
even write down an statement (e.g., NP or QMA) that encodes
consistency?  Overcoming these difficulties constitutes the bulk of
our technical work. We provide a more detailed discussion on the
technical contribution of our work in \Cref{sec:technical}.

We remark that understanding our protocol requires only limited knowledge of
quantum computation. Thanks to the composition theorems for (stand-alone) simulation-secure
quantum protocols~\cite{HSS15}, much of our protocol can be viewed as a
{\em classical} protocol in the
(unbounded simulation) OT-hybrid model. The only quantumness resides in the
instantiation of the OT hybrid
with %
\cite{C:BBCS91}.%

We notice that just as in \cite{BB84,C:BBCS91}, the honest execution
of our protocols does not need strong quantum computational power,
since one only needs to create, send and measure ``BB84'' states,
which can be performed with current quantum technology. \footnote{A
  BB84 state is a single-qubit state that is chosen uniformly at
  random from $\{\ket{0},\ket{1},\ket{+},\ket{-}\}$. Alternatively, it
  can be prepared by computing $H^hX^x \ket{0}$ where $X$ is the
  bit-flip gate, $H$ is the Hadamard gate, and $h,x\in\{0,1\}$ are
  random bits.}  Most notably, creating the states does not involve
creating or maintaining long-range correlations between qubits.

In turn, plugging our OT protocol into the protocols
of \cite{C:IshPraSah08,EC:Unruh10,DNS12,EC:DGJMS20} (and using the sequential
composition theorem~\cite{HSS15}) gives us secure two-party computation
and multi-party computation (with a dishonest majority) protocols, even for
quantum channels.

\begin{theorem}[Informal]
  Assuming that post-quantum one-way functions exist and quantum communication
  is possible, for every
  classical two-party and multi-party functionality $\func$, there is
  a quantum protocol in the plain model that is simulation-secure
  against malicious quantum polynomial-time adversaries.
  Under the same assumptions, there is a quantum two-party and multi-party
  protocol for any
  quantum circuit $Q$.
\end{theorem}

Finally, we note that our OT protocol runs in $\mathsf{poly}(\secpar)$ number
of rounds, where $\secpar$ is a security parameter, and that is only because
of the zero-knowledge proof. Watrous' ZK proof system~\cite{Watrous09} involves
repeating a classical ZK proof (such as that graph coloring ZK proof~\cite{C:GolMicWig86}
or the Hamiltonicity proof~\cite{Blum86}) {\em sequentially}. A recent work
of Bitansky and Shmueli~\cite{STOC:BitShm20} for the first time constructs a
{\em constant-round} quantum ZK protocol (using only classical resources) but
they rely on a strong assumption, namely learning with errors, which
 does not
live in minicrypt. Nevertheless, in the common random string (CRS) model,
we can instantiate the zero-knowledge protocol using a WI protocol and a Pseudo-Random Generator (PRG) with additive $\lambda$ bit stretch as follows: To prove a statement $x$, the prover proves using the WI protocol that either $x$ is in the language or the common random string is in the image of the PRG. To simulate a proof, the simulator samples the CRS as a random image of the PRG, and proves using the WI protocol that it belongs to the image in a straight-line. 
Moreover, this modification allows us to achieve {\em straight-line simulators}, leading to {\em universally-composable} (UC) security~\cite{Can01}.
Therefore, this
modification would give us the following statement.  

\begin{theorem}[Informal]
  Constant-round oblivious transfer protocols in the common random string (CRS) model that are  UC-simulation-secure
  against malicious quantum poly-time adversaries
  exist assuming that post-quantum one-way functions exist and that quantum communication is possible.
\end{theorem}
Plugging the above UC-simulation-secure OT into the protocol of~\cite{C:IshPraSah08} gives constant-round multi-party computation protocols for classical computation in the common random string model that are UC-simulation-secure against malicious quantum poly-time adversaries.
\rachel{Rachel: 11/29 added the above sentence for constant round pq MPC.}

\rnotem{if ZK in the CRS model from OWF works, which is simply PRG in
  CRS + WI, we can write two theorems, one in the plain model and
  anohter in the CRS model}

\agnotem{Maybe we should sell our result a bit more here}

\paragraph{Going Below MiniQCrypt?}
We notice that all of the primitives that we implement in our work {\em cannot} be implemented unconditionally, even in the quantum setting~\cite{Mayers97,LoChau97,Lo97,CGS16}. Basing their construction on \pqOWF{}s seems to be the next best thing, but it does leave with the intriguing question if they could be based on weaker assumptions.
More concretely, assume a world with quantum communication as we do in this paper. Does the existence of quantum OT protocols imply the existence of \pqOWF{}s? Or, does a weaker {\em quantum} notion of one-way functions suffice? We leave the exploration of other possible cryptographic worlds below MiniQCrypt to future work.

\paragraph{Other Related Work.}
Inspired by the quantum OT protocol~\cite{C:BBCS91}, a family of
primitives, named \emph{$k$-bit cut-and-choose}, has been shown to be
sufficient to realize OT statistically by quantum
protocols~\cite{FKSZZ13,DFLS16} which is provably impossible by
classical protocols alone~\cite{MPR10}. These offer further examples
demonstrating the power of quantum cryptographic protocols.

There has also been extensive effort on designing quantum protocols OT
and the closely related primitive of \emph{one-time-memories} under
\emph{physical} rather than \emph{computational} assumptions, such as
the bounded-storage model, noisy-storage model, and isolated-qubit
model, which restrict the quantum memory or admissible operations of
the
adversary~\cite{Salvail98,Liu14_itcs,Liu14_crypto,DFR+07,DFSS08,KWW12}. They
provide important alternatives, but the composability of these
protocols are not well understood. Meanwhile, there is strengthening
on the impossibility for quantum protocols to realize secure
computation statistically from scratch~\cite{BCS12,SSS15}.

Finally, we note that there exist classical protocols for two-party
and multi-party computation that are quantum-secure assuming strong
assumptions such as post-quantum dense encryption and superpolynomial
quantum hardness of the learning-with-errors
problem~\cite{HSS15,LN11,ABGKM20}. And prior to the result
in~\cite{EC:DGJMS20}, there is a long line of work on secure multi-party
\emph{quantum} computation (Cf.~\cite{CGS02,BCGHS06,DNS10,DNS12}).

\subsection{Technical Overview}
\label{sec:technical}

We give an overview of our construction of post-quantum OT protocol in
the plain model from post-quantum one-way functions.
In this overview, we assume some familiarity with post-quantum MPC in the
stand-alone, sequential composition, and UC models, and basic
functionalities such as $\fot$ and $\fcom$.  We will also consider
{\em parallel versions} of them, denoted as $\fpot$ and $\fsocom$. The
parallel OT functionality $\fpot$ enables the sender to send some
polynomial number of pairs of strings $\smallset{s^i_0, s^i_1}_i$ and the receiver to choose one per pair to obtain
$s^i_{c_i}$ in
parallel. The commitment with selective opening functionality
$\fsocom$ enables a sender to commit to a string $m$ while hiding it,
and a receiver to request opening of a subset of bits at locations
$T \subseteq[|m|]$ and obtain $m_T = (m_i)_{i \in T}$. We
refer the reader to Section~\ref{sec:model} for formal definitions of
these functionalities.

\paragraph{BBCS OT in the $\fsocom$-Hybrid Model.} We start by
describing the
quantum OT protocol of \cite{C:BBCS91} in the $\fsocom$ hybrid model.

\begin{mdframed}[style=figstyle,innerleftmargin=10pt,innerrightmargin=10pt]
  \small
{\bf BBCS OT protocol:} The sender $\otS$ has strings $s_0,s_1 \in \zo^\ell$, the receiver
$\otR$ has a choice bit $c \in \zo$.
\begin{enumerate}
\item {\bf Preamble.} $\otS$ sends $n \gg \ell$ BB94 qubits
  $\ket{x^A}_{ \theta^A}$ prepared using random bits
  $ x^A\in_R \bits^n$ and random basis
  $ \theta^A\in_R \base^n$.

  $\otR$ measures these qubits in randomly chosen bases
  $ \theta^B\in_R \base^n$ and commits to the measured bits
  together with the choice of the bases, that is $\smallset{ \theta_i^B,  x_i^B}_i$,
  using
  $\fsocom$.

\item {\bf Cut and Choose.}  $\otS$ requests to open a random
  subset $T$ of locations, of size say $n/2$, and
  gets $\{  \theta_i^B,  x_i^B\}_{i\in T}$ from
  $\fsocom$.

  Importantly, it aborts if for any $i$
  $\theta^B_i = \theta^A_i$ but
  $ x^B_i\neq  x^A_i$. Roughly speaking, this is because it's
  an indication that the receiver has not reported honest
  measurement outcomes.

\item {\bf Partition Index Set.}  $\otS$ reveals
  ${\theta}^A_{\bar T}$ for the unchecked locations $\bar T$. $\otR$ partitions $\bar T$ into a subset of locations where it
  measured in the same bases as the sender
  $I_c := \{i \in \bar T: {\theta}_i^A ={\theta}_i^B\}$ and
  the rest $I_{1-c} := \bar T - I_c$, and sends $(I_0, I_1)$ to the sender.

\item {\bf Secret Transferring}. $\otS$ hides the two strings
  $s_i$ for $i = 0, 1$ using randomness extracted from
  $ x^A_{I_i}$ via a universal hash function $f$ and sends
  $m_i := s_i \oplus f( x^A_{I_i})$, from which $\otR$
  recovers $s := m_c \oplus f( x^B_{I_c})$.

\end{enumerate}
\end{mdframed}

Correctness follows from that for every $i\in I_c$, $\theta_i^A = \theta_i^B$ and
  $x^A_{I_c} = x^B_{I_c}$, hence the receiver decodes
  $s_c$ correctly. 

The security of the BBCS OT protocol relies crucially on two important
properties of the $\fsocom$ commitments, namely extractability and
equivocability, which any protocol implementing the $\fsocom$ functionality must
satisfy.

\medskip
\noindent\highlight{Equivocability:} To show the receiver's privacy, we need
  to efficiently simulate the execution with a malicious sender
  $\otS^*$ without knowing the choice bit $c$ and extract both
  sender's strings $s_0, s_1$. To do so, the simulator $\otSimS$ would
  like to measure at these unchecked locations $\bar T$ using exactly
  the same bases $\theta^A_{\bar T}$ as $\otS^*$ sends in Step 3. In
  an honest execution, this is impossible as the receiver must commit
  to its bases $\theta^B$ and pass the cut-and-choose step. However,
  in simulation, this can be done by invoking the equivocability of
  $\fsocom$. In particular, $\otSimS$ can {\em simulate} the
  receiver's commitments in
  the preamble phase without committing to any value. When it is
  challenged to open locations at $T$, it measures qubits at $T$ in
  random bases, and {\em equivocates} commitments at $T$ to the
  measured outcomes and bases. Only after $\otS^*$ reveals its bases
  $\theta^A_{\bar T}$ for the unchecked locations, does $\otSimS$
  measure qubits at $\bar T$ in exactly these bases. This ensures that
  it learns both $x^A_{I_0}$ and $x^A_{I_1}$ and hence can recover
  both $s_0$ and $s_1$.
  
  \smallskip
\noindent\highlight{Extractability:} To show the sender's privacy, we need to
  efficiently extract the choice bit $c$ from a malicious receiver
  $\otR^*$ and simulate the sender's messages using only
  $s_c$. To do so, the simulator $\otSimR$ needs to extract
  efficiently from the $\fsocom$ commitments all the bases $\theta^B$,
  so that, later given $I_0, I_1$ it can figure out which subset $I_c$
  contains more locations $i$ where the bases match
  $\theta^B_i=\theta^A_i$, and use the index of that set as the
  extracted choice bit. Observe that it is important that extraction
  does not ``disturb'' the quantum state of $\otR^*$ at all, so that
  $\otSimR$ can continue simulation with $\otR^*$. This is easily
  achieved using $\fsocom$ as extraction is done in a straight-line
  fashion, but challenging to achieve in the plain model as rewinding
  a quantum adversary is tricky. Indeed, the argument of knowledge
  protocol of~\cite{EC:Unruh12} can extract a witness but disturbs the
  state of the quantum adversary due to measurement. To the best of
  our knowledge, such strong extractable commitment is only known
  assuming post-quantum FHE in the
  plain model~\cite{STOC:BitShm20,AnaPla19}\agnotem{Is the second citation the one you meant?}
  using non-black-box simulation techniques, or assuming public key
  encryption in the CRS model.
\medskip

It turns out that equivocability {\em can} be achieved
using zero-knowledge protocols, which gives a post-quantum OT protocol
with an inefficient simulator $\otSimR$ against malicious receivers
(and efficient $\otSimS$). Our main technical contribution lies in
achieving efficient extractability while assuming only post-quantum one-way
functions. In particular, we will use the OT with unbounded simulation
as a tool for this. We proceed to describing these steps in more detail.

\paragraph{Achieving Equivocability Using Zero-Knowledge.} The idea is
to let the committer commit $c = \com(\mu; \rho)$ to a string
$\mu \in \zo^n$ using any statistically binding computationally hiding commitment scheme
$\com$ whose decommitment can be verified classically, for instance, Naor's commitment scheme~\cite{C:Naor89} from post-quantum one-way functions. For now in this overview, think of $\com$ as non-interactive. 
(Jumping ahead, later we will also instantiate this commitment with a multi-round extractable commitment scheme that we construct.)

Any computationally hiding commitment can be simulated by simply committing to zero,
$\tilde c = \com(0; \rho)$. The question is how to equivocate
$\tilde c$ to any string $\mu'$ later in the decommitment phase.
With a
post-quantum ZK protocol, instead of asking the committer to reveal
its randomness $\rho$ which would statistically bind $\tilde c$ to the
zero string, we can ask the committer to send $\mu'$ and give a
zero-knowledge proof that $\tilde c$ indeed commits to $\mu'$. As such,
the simulator can cheat and successfully open to any value $\mu'$ by
simulating the zero-knowledge argument to the receiver.

\begin{mdframed}[style=figstyle,innerleftmargin=10pt,innerrightmargin=10pt]
  \small
{\bf Equivocable Commitment:} The sender $\comS$ has a string $\mu \in \zo^n$, the receiver
$\comR$ has a subset $T \subseteq [n]$.
\begin{enumerate}
\item {\bf Commit Phase.} $\comS$ commits to $\mu$ using a statistically
  binding commitment scheme $\com$ using randomness $\rho$. Let $c$ be
  the produced commitment.

  \Note{Simulation against malicious receivers commits to $0^n$. 
    Simulation against malicious senders is inefficient to extract
    $\mu$ by brute force.}

\item {\bf Decommit Phase.}  Upon $\comR$ requesting to open a subset
  $T$ of locations, $\comS$ sends $\mu'$ and gives a single zero
  knowledge argument that $c$ commits to $\mu$ such that $\mu' = \mu_T$.

  \Note{To equivocate to $\mu' \neq \mu_T$, the simulator sends $\mu'$ and
    simulates the zero-knowledge argument (of the false statement).}
\end{enumerate}
\end{mdframed}
The above commitment protocol implements $\fsocom$ with efficient
simulation against malicious receivers, but inefficient simulation
against malicious senders. Plugging it into BBCS OT protocol, we
obtain the following corollary:
\begin{corollary}[Informal]
  Assume post-quantum one-way functions. In the plain model, there is:
  \begin{itemize}
  \item a protocol that securely implements the OT functionality $\fot$, and
  \item a protocol that securely implements the parallel OT
    functionality $\fpot$,
\end{itemize}
in the sequential composition setting, and with efficient
simulation against malicious senders but inefficient simulation
against malicious receivers. \rnotem{CRS model}
\end{corollary}
The second bullet requires some additional steps, as parallel
composition does not automatically apply in the stand-alone  (as opposed to UC) setting
(e.g., the ZK protocol of~\cite{Watrous09} is not
simulatable in parallel due to rewinding). Instead, we first observe that the BBCS OT
UC-implements $\fot$ in the $\fsocom$ hybrid model, and hence parallel
invocation of BBCS OT UC-implements $\fpot$ in the $\fsocom$ hybrid
model. Note that parallel invocation of BBCS OT invokes $\fsocom$ in
parallel, which in fact can be merged into a single invocation to
$\fsocom$. 
 Therefore,
plugging in the above commitment protocol gives an OT protocol that
implements $\fpot$. In particular, digging deeper into the protocol, this ensures that
we are invoking a {\em single} ZK protocol for all the parallel copies of the parallel OT, binding
the executions together.

\paragraph{Achieving Extractability Using OT with Unbounded
  Simulation.} Interestingly, we show that OT with (even 2-sided)
unbounded simulation plus zero-knowledge is sufficient for
constructing extractable commitments, which when combined with
zero-knowlege again as above gives an implementation of $\fsocom$ in
the sequential composition setting in the plain model.

The initial idea is to convert the power of simulation into the power
of extraction via two-party computation, and
sketched below. \vnotem{The ref here is probably wrong.}\rnotem{let's just not cite.}
\begin{mdframed}[style=figstyle,innerleftmargin=10pt,innerrightmargin=10pt]
  \small
{\bf Initial Idea for Extractable Commitment:} The sender $\comS$ has $\mu \in \zo^n$.

\begin{enumerate}
\item {\bf Trapdoor setup:}
  The receiver $\comR$ sends a commitment $c$ of a statistically
  binding commitment scheme $\com$, and gives a zero-knowledge proof
  that $c$ commits to 0.

\item {\bf Conditional Disclosure of Secret (CDS):}
  $\comS$ and $\comR$ run a two-party computation protocol implementing
  the CDS functionality $\fcds$ for the language
  $\Lang_\com = \smallset{(c', b') : \exists r' \text{ s.t. } c' =
    \com(b'; r')}$, where the CDS functionality $\fcds$ for $\Lang_\com$ is defined as below:
  \begin{align*}
    \fcds \ :&  \text{ Sender input } (x, \mu), \text{ Receiver
    input } w\\
  & \text{ Sender has no output}, \text{Receiver outputs } x \mbox{ and } \mu'
                                     =\begin{cases}
                                       \mu & \mbox{if } \Rel_{\Lang_\com}(x, w) = 1\\
                                       \perp & \text{otherwise}
                                     \end{cases}
  \end{align*}
  $\comS$ acts as the CDS sender using
  input $(x=(c, 1), \mu)$ while $\comR$ acts as the CDS receiver using
  witness $w = 0$.
\end{enumerate}
\end{mdframed}
It may seem paradoxical that we try to implement commitments using the
much more powerful tool of two-party computation. The {\em key observation} is
that the hiding and extractability of the above commitment protocol
 only relies on the {\em input-indistinguishability property} of the CDS
protocol, which is {\em implied by unbounded simulation}.
\begin{itemize}
\item \highlight{Hiding:} A commitment to $\mu$ can be simulated by
  simply commiting to $0^n$ honestly, that is, using $(x=(c, 1), 0^n)$
  as the input to the CDS. The simulation is
  indistinguishable as the soundness of ZK argument guarantees that $c$
  must be a commitment to $0$ and hence the CDS statement $(c, 1)$ is
  false and should always produce $\mu' = \bot$. Therefore, the
  unbounded-simulation security of the CDS protocol implies that it is
  indistinguishable to switch the sender's input from $\mu$ to $0^n$.

\item \highlight{Extraction:} To efficiently extract from a malicious
  sender $\comS^*$, the idea (which however suffers from a problem
  described below) is to let the simulator-extractor $\comSimS$ set
  up a trapdoor by committing to 1 (instead of 0) and simulate the ZK
  argument; it can then use the decommitment (call it $r$) to 1 as a valid witness
  to obtain the committed value from the output of the CDS
  protocol. Here, the unbounded-simulation security of CDS again
  implies that interaction with an honest receiver who uses $w = 0$ is
  indistinguishable from that with $\comSimS$ who uses $w = r$ as
  $\comS^*$ receives no output via CDS.
\end{itemize}
The advantage of CDS with unbounded simulation is that it can be
implemented using OT with unbounded simulation: Following the work
of~\cite{STOC:Kilian88,C:IshPraSah08,EC:Unruh10}, post-quantum MPC protocols exist in the
$\fot$-hybrid model, and instantiating them with the unbounded-simulation
OT yields unbounded simulation MPC and therefore CDS.

\subparagraph{NP-Verifiability and the Lack of It.} Unfortunately, the above
attempt  has several problems: how do we show that the commitment is binding?
how to decommit? and how to guarantee that the extracted value agrees
with the value that can be decommitted to? %
We can achieve binding by having the sender additionally commit to
$\mu$ using a statistically binding commitment scheme $\com$, and send
the corresponding decommitment in the decommitment phase. However, to
guarantee that the extractor would extract the same string $\mu$ from
CDS, we need a way to verify that the same $\mu$ is indeed used by the CDS
sender.
Towards this, we formalize a verifiability property of a CDS protocol:

\smallskip\noindent{\em A CDS protocol is verifiable if}\vspace{-1mm}
\begin{itemize}
\item The honest CDS sender $\cdsS$ additionally outputs $(x,\mu)$ and a
  ``proof'' $\pi$ (on a special output tape) at the end of the execution.
\item There is an efficient {\em classical} verification algorithm
  $\Ver(\tau, x,\mu, \pi)$ that verifies the proof, w.r.t.\ the
  transcript $\tau$ of the {\em classical} messages exchanged in the CDS protocol.
\item \highlight{Binding:} No malicious sender $\cdsS^*$ after
  interacting with an honest receiver $\cdsR(w)$ can output
  $(x,\mu,\pi)$, such that the following holds simultaneously: (a)
  $\Ver(\tau, x, \mu, \pi) = 1$, (b) $\cdsR$ did not abort, and (c)
  $\cdsR$ outputs $\mu'$ inconsistent with the inputs $(x,\mu)$ and $w$,
  that is,
$      \mu' \ne
      \begin{cases}
        \mu & \text{ if } \Rel_\Lang(x, w) = 1  \\
        \bot & \text{ otherwise }
      \end{cases}
$
\end{itemize}

We observe first that classical protocols with perfect correctness have
verifiability for free: The proof $\pi$ is simply the sender's
random coins $r$, and the verification checks if the honest sender
algorithm with input $(x, \mu)$ and random coins $r$ produces the same
messages as in the transcript $\tau$. If so, perfect correctness
guarantees that the output of the receiver must be consistent with
$x, \mu$. However, verifiability cannot be taken for granted in the
$\fot$ hybrid model or in the quantum setting. In the $\fot$ hybrid
model, it is difficult to write down an NP-statement that captures
consistency as the OT input is {\em not} contained in the protocol
transcript and is unconstrained by it. In the quantum setting,
protocols use quantum communication, and consistency cannot be
expressed as an NP-statement.  %
Take the BBCS protocol as an example, the OT receiver receives from
the sender $\ell$ qubits and measures them locally; there is no way to
"verify" this step in NP.

\paragraph{Implementing Verifiable CDS.} To overcome the above
challenge, we implement a verifiable CDS protocol in the $\fpot$
hybrid model assuming only post-quantum one-way functions. We develop
this protocol in a few steps below.

Let's start by understanding why the standard two-party comptuation
protocol is not verifiable. The protocol proceeds as follows:
First, the sender $\cdsS$ locally garbles a circuit computing the
following function into $\hat G$ with labels
$\smallset{\ell^j_b}_{j \in [m], b\in\zo}$ where $m = |w|$.
\begin{align}
  G_{x,\mu}(w) =  \mu' = \begin{cases}
    \mu & \text{ if } \Rel_\Lang(x, w) = 1  \\
    \bot & \text{ otherwise }
  \end{cases}
  \label{eq:G}
\end{align}
Second, $\cdsS$ sends the pairs of labels $\smallset{\ell^j_0, \ell^j_1}_j$
via $\fpot$. The receiver $\cdsR$ on the other hand chooses $\smallset{w_j}_{j}$ to
obtain $\smallset{\tilde \ell^j_{w_j}}_j$, and  evaluates
$\hat G$ with these labels to obtain $\mu'$.  This protocol is not
NP-verifiable because consistency between the labels of the garbled
circuit and the sender's inputs to $\fpot$ cannot be expressed as a NP
statement.

To fix the problem, we devise a way for the receiver to verify the OT
sender's strings. Let $\cdsS$ additionally commit to all the labels
$\smallset{c_b^j = \com(\ell^j_b ; r^j_b)}_{j,b}$ and the message
$c = \com(\mu; r)$ and prove in ZK that $\hat G$ is consistent with the labels and message committed in 
the commitments, as well as the statement $x$. Moreover, the sender sends both
the labels and decommitments
$\smallset{(\ell^j_0, r^j_0),(\ell^j_1, r^j_1)}_j$ via $\fpot$.
The receiver after receiving
$\smallset{\tilde \ell^j_{w_j}, \tilde r^j_{w_j}}_j$ can now verify 
their correctness by verifying the decommitment w.r.t.\ $c_{w_j}^j$,
and aborts if verification fails. This gives the following new
protocol:

\begin{mdframed}[style=figstyle,innerleftmargin=10pt,innerrightmargin=10pt]
  \small
{\bf A Verifiable but Insecure CDS Protocol:} The sender $\cdsS$ has
$(x, \mu)$ and the receiver $\cdsR$ has $w$.

\begin{enumerate}
\item {\bf Sender's Local Preparation:} $\cdsS$ generate a
  garbled circuits $\hat G$ for the
  circuit computing $G_{x,\mu}$ (Equation~\eqref{eq:G}), with labels
  $\smallset{\ell^{i,j}_b}_{j, b}$. Moreover, it generates commitments $c = \com(\mu, r)$ and
$c^{j}_b = \com(\ell^{j}_b; r^{j}_b)$ for every $j,b$.

\item {\bf OT:} $\cdsS$ and $\cdsR$ invoke $\fpot$. For every $j$, the
  sender sends $(\ell^{j}_0, r^{j}_0), (\ell^{j}_1, r^{j}_1)$, and the
  receiver chooses $w_j$ and obtains
  $(\tilde \ell^{j}_{w_{j}}, \tilde r^{j}_{w_{j}})$.

\item {\bf Send Garbled Circuit and Commitments:} $\cdsS$ sends
  $\hat G$, $c$, and $\smallset{c^j_b}_{j,b}$ and proves via a ZK
  protocol that they are all generated consistently w.r.t.\ each other
  and $x$.

\item {\bf Receiver's Checks:} $\cdsR$ aborts if ZK is not accepting,
  or if for some $j$,
  $c^j_{w_j} \ne \com(\tilde \ell^{j}_{w_{j}}, \tilde
  r^{j}_{w_{j}})$. Otherwise, it evaluates $\hat G$ with the labels
  and obtain $\mu' = G_{x,\mu}(w)$.
\end{enumerate}
\end{mdframed}
We argue that this protocol is NP-verifiable. The sender's proof
is simply the decommitment $r$ of $c$, and
$\Ver(\tau, (x, \mu), r) = 1$ iff $r$ is a valid decommitment to $\mu$ of the commitment $c$ contained
in the transcript $\tau$. To show the binding property, consider an
interaction between a cheating sender $\cdsS^*$ and
$\cdsR(w)$. Suppose $\cdsR$ does not abort, it means that 1) the ZK
argument is accepting and hence $\hat G$ must be consistent with
$x, \smallset{c_b^j}, c$, and 2) the receiver obtains the labels
committed in $c_{w_j}^j$'s. Therefore, evaluating the garbled circuit
with these labels must produce $\mu'= G_{x,\mu}(w)$ for the $\mu$
committed to in $c$.

Unfortunately, the checks that the receiver performs render the protocol
insecure. A malicious sender $\comS^*$ can launch the so-called
selective abort attack to learn information of $w$. For instance, to
test if $w_1 = 0$ or not, it replaces $\ell^1_0$ with
zeros.  If $w_1 = 0$ the honest receiver would abort; otherwise, it
proceeds normally.

\subparagraph{The Final Protocol} To circumvent the selective abort
attack, we need a way to check the validity of sender's strings that
is independent of $w$. Our idea is to use a variant of cut-and-choose. Let $\cdsS$
create $2\lambda$ copies of garbled circuits and commitments to their
labels, $\smallset{\hat G^i}_{i \in [2\lambda]}$ and
$\smallset{c^{i,j}_b = \com(\ell^{i,j}_b; r^{i,j}_b)}_{i,j,b}$ and
prove via a ZK protocol that they are all correctly generated w.r.t.\
the same $c$ and $x$. Again, $\cdsS$ sends the labels and decommitment
via $\fpot$, but $\cdsR$ does not choose $w$ universally in all
copies. Instead, it secretly samples a random subset
$\Lambda \in [2\lambda]$ by including each $i$ with probability 1/2; for
copy $i \in \Lambda$, it chooses random string $s^i \gets \zo^m$ and obtains $\smallset{\tilde \ell^{i,j}_{s^i_j}, \tilde r^{i,j}_{s^i_j}}_j$,
whereas for copy $i \not \in \Lambda$, it choose $w$ and obtains $\smallset{\tilde \ell^{i,j}_{w_j}, \tilde r^{i,j}_{w_j}}_j$.  Now, in the
checking step, $\cdsR$ only verifies the validity of $\smallset{\tilde \ell^{i,j}_{s^i_j}, \tilde r^{i,j}_{s^i_j}}_{i \in \Lambda, j}$  received in copies in $\Lambda$. Since the check is
now completely independent of $w$, it circumvents the selective abort attack.

Furthermore, NP-verifiability still holds. The key point is that if
the decommitments $\cdsR$ receives in copies in $\Lambda$ are all
valid, with overwhelming probability, the number of {\em bad copies}
where the OT sender's strings are not completely valid is bounded by
$\lambda/4$. Hence, there must exist a copy $i \not\in \Lambda$ where
$\cdsR$ receives the right labels $\ell^{i,j}_{w_j}$ committed to in
$c^{i,j}_{w_j}$. $\cdsR$ can then evaluate $\hat G^i$ to obtain
$\mu'$. By the same argument as above, $\mu'$ must be consistent with
the $(x,\mu)$ and $w$, for $\mu$ committed in $c$, and
NP-verifiability follows. The final protocol is described in Figure~\ref{fig:cdsprot}.

\rachelm{Another angle is what is the right primitive for defining minicrypt for qcrypt? Is it classical OWF or something weaker?}
\vnotem{Good question. I think it depends on what we can prove :) We could have miniQcrypt (quantum OWF with quantum input/output -- is this well-defined -- or microQcrypt (quantum OWF with classical input/output) or nanoQcrypt (classical OWF)}\agnotem{I am not sure if we should not invert the definitions of mini/micro/nano... Intuitively, it seems that it is easier to have a OWF with quantum input/outputs... The point is that they might be incomparable}

\subsubsection{Organization of the Paper.}
We review the quantum stand-alone security model introduced by \cite{HSS15} in \Cref{sec:model}. In section~\Cref{sec:parallel-ot}, we construct a quantum parallel-OT protocol with one-sided, unbounded simulation. In more detail, we review in \Cref{sec:qot}  the quantum OT protocol from \cite{C:BBCS91} based on ideal commitments with selective opening security. Then in \Cref{sec:parallel-composition}, we show how to boost it to construct a {\em parallel} OT protocol from the same assumptions. And finally, we provide a classical implementation of the commitment scheme with selective opening security in \Cref{sec:socom} which gives us ideal/real security except with unbounded receiver simulation. This result will be fed into our main technical contribution in \Cref{sec:ecom} where we show how to construct extractable commitments from unbounded-simulation parallel-OT. In \Cref{sec:cdsprot}, we show how to construct (the intermediate primitive of) CDS from parallel-OT and one-way functions, and then in \Cref{sec:cds-to-extcom} we construct extractable commitments from CDS.
Finally, in \Cref{sec:mpc} we lift our results to achieve quantum protocols for multi-party (quantum) computation from one-way functions.

Throughout work, we assume familiarity with  basic notions in quantum computation and cryptography. For completeness,  we provide a brief review on the relevant concepts in~\Cref{sec:prelim}.

\section{Quantum Stand-alone Security Model}
\label{sec:model}

We adopt the quantum stand-alone security model from the work of Hallgren, Smith and Song~\cite{HSS15}, tailored to the two-party setting.

Let $\func$ denote a \emph{functionality}, which is a
classical \vnotem{``or quantum'': do we need it? aren't all our $\func$ classical?} interactive machine specifying the instructions
to realize a cryptographic task. A two-party protocol $\Pi$ consists
of a pair of quantum interactive machines $(A,B)$. We call a protocol
{\em efficient} if $A$ and $B$ are both quantum poly-time machines. If we want to
emphasize that a protocol is classical, i.e., all computation and all
messages exchanged are classical, we then use lower-case letters
(e.g., $\pi$). Finally, an adversary $\adv$ is another quantum interactive
machine that intends to attack a protocol.

When a protocol $\Pi=(A,B)$ is executed under the presence of an
adversary $\adv$, the state registers are initialized by a security
parameter $1^\secpar$ and a joint quantum state $\sigma_\secpar$.
Adversary $\adv$ gets activated first, and may either \textbf{deliver}
a message, i.e., instructing some party to read the proper segment of
the network register, or \textbf{corrupt} a party. We assume all
registers are authenticated so that $\adv$ cannot modify them, but
otherwise $\adv$ can schedule the messages to be delivered in any
arbitrary way. If $\adv$ corrupts a party, the party passes all of its
internal state to $\adv$ and follows the instructions of $\adv$. Any
other party, once receiving a message from $\adv$, gets activated and
runs its machine. At the end of one round, some message is generated
on the network register. Adversary $\adv$ is activated again and
controls message delivery. At some round, the party generates some
output and terminates.

We view $\Pi$ and $\adv$ as a whole and model the composed system as
another QIM, call it $\mac_{\Pi, \adv}$. Then executing $\Pi$ in the
presence of $\adv$ is just running $\mac_{\Pi, \adv}$ on some input
state, which may be entangled with a reference system available to a
distighuisher.

\paragraph{Protocol emulation and secure realization of a
  functionality.} A secure protocol is supposed to ``emulate'' an
idealized protocol. Consider two protocols $\Pi$ and $\Gamma$, and let
$\mac_{\Pi,\adv}$ be the composed machine of $\Pi$ and an adversary
$\adv$, and $\mac_{\Gamma,\simulator}$ be that of $\Gamma$ and another
adversary $\simulator$. Informally, $\Pi$ emulates $\Gamma$ if the two
machines $\mac_{\Pi,\adv}$ and $\mac_{\Gamma,\simulator}$ are
indistinguishable.

Given the general form of protocol emulation, it is of particular
interest to emulate the so-called \emph{ideal-world} protocol
$\tilde \Pi_\func$ for a functionality $\func$ which captures the
security properties we desire. In this protocol, two (dummy) parties
$\tilde A$ and $\tilde B$ have access to an additional ``trusted''
party that implements $\func$. We abuse notation and call the trusted
party $\func$ too. Basically $\tilde A$ and $\tilde B$ invoke $\func$
with their inputs, and then $\func$ runs on the inputs and sends the
respective outputs back to $\tilde A$ and $\tilde B$. An execution of
$\tilde \Pi$ with an adversary $\simulator$ is as before, except that
$\func$ cannot be corrupted. We denote the composed machine of $\func$
and $\tilde \Pi_\func$ as $\mac_{\func, \simulator}$.

\begin{definition}[Computationally Quantum-Stand-Alone Emulation]\label{def:cqsa}
  Let
  $\Pi$ and $\Gamma$ be two poly-time protocols. We say $\Pi$
  \emph{computationally quantum-stand-alone} (\cqsa) emulates
  $\Gamma$, if for any poly-time QIM $\adv$ there exists a poly-time
  QIM $\calS$ such that $\mac_{\Pi,\adv} \qc
    \mac_{\Gamma,\simulator} $.
\label{def:qcsa}
\end{definition}

\begin{definition}[\cqsa~Realization of a Functionality] Let $\func$
  be a poly-time two-party functionality and $\Pi$ be a poly-time
  two-party protocol. We say $\Pi$ {\em computationally
    quantum-stand-alone} realizes $\func$, if $\Pi$ \cqsa~emulates
  $\tilde \Pi_\func$. Namely, for any poly-time $\adv$, there is a
  poly-time $\simulator$ such that $\mac_{\Pi, \adv} \qc \mac_{\func,
    \simulator}$.
\label{def:qsae}
\end{definition}

\vspace{-3ex}

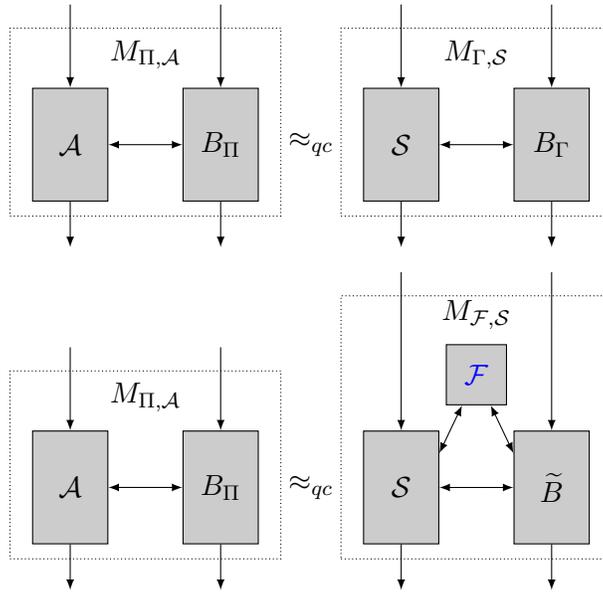
\begin{figure}[ht!]
    \begin{center}
    \begin{tikzpicture}[scale=0.38,
        turn/.style={draw, minimum height=15mm, minimum width=10mm,
          fill = ChannelColor, text=ChannelTextColor},
        smallturn/.style={draw, minimum height=8mm, minimum width=8mm,
          fill = ChannelColor, text=ChannelTextColor},
        >=latex]
      \node (A) at (-8,0) [turn] {$\adv$};
      \node (B) [right of = A, xshift = 1cm, turn] {$B_\Pi$};
      \node (ina) [above of=A, yshift=1cm, minimum width = 0mm] {};
      \node (inb) [above of=B, yshift=1cm, minimum width = 0mm] {};
      \node (outa) [below of=A, yshift=-0.5cm, minimum width = 0mm] {};
      \node (outb) [below of=B, yshift=-0.5cm, minimum width = 0mm]
      {};
      \node (M) [above of = A, xshift = 1cm, yshift=2mm, minimum width = 0mm] {$M_{\Pi, \adv}$};
      \node (ML) [left of = A, xshift = 2cm, yshift = 0.3cm, draw, densely dotted, minimum height= 25 mm,
        minimum width = 3.6cm] {};
      \draw[<->] (A.east) -- (B.west) node {};
      \draw[->] (ina.south) -- (A.north) node {};
      \draw[->] (inb.south) -- (B.north) node {};
      \draw[->] (A.south) -- (outa.north) node {};
      \draw[->] (B.south) -- (outb.north) node {};

      \node (QC) [right of = B, xshift=.2cm] {$\qc$};

      \node (A2) [right of =QC, xshift = .2cm, turn] {$\simulator$};
      \node (B2) [right of = A2, xshift = 1cm, turn] {$B_\Gamma$};
      \node (ina2) [above of=A2, yshift=1cm, minimum width = 0mm] {};
      \node (inb2) [above of=B2, yshift=1cm, minimum width = 0mm] {};
      \node (outa2) [below of=A2, yshift=-0.5cm, minimum width = 0mm] {};
      \node (outb2) [below of=B2, yshift=-0.5cm, minimum width = 0mm]
      {};
      \node (M2) [above of = A2, xshift = 1cm, yshift=2mm, minimum width = 0mm] {$M_{\Gamma, \simulator}$};
      \node (ML2) [left of = A2, xshift = 2cm, yshift = 0.3cm, draw, densely dotted, minimum height= 25 mm,
        minimum width = 3.6cm] {};
      \draw[<->] (A2.east) -- (B2.west) node {};
      \draw[->] (ina2.south) -- (A2.north) node {};
      \draw[->] (inb2.south) -- (B2.north) node {};
      \draw[->] (A2.south) -- (outa2.north) node {};
      \draw[->] (B2.south) -- (outb2.north) node {};

      \node (A3) at (-8,-12) [turn] {$\adv$};
      \node
      (B3) [right of = A3, xshift = 1cm, turn] {$B_\Pi$}; \node (ina3)
      [above of=A3, yshift=1cm, minimum width = 0mm] {}; \node (inb3)
      [above of=B3, yshift=1cm, minimum width = 0mm] {}; \node (outa3)
      [below of=A3, yshift=-0.5cm, minimum width = 0mm] {}; \node
      (outb3) [below of=B3, yshift=-0.5cm, minimum width = 0mm] {};
      \node (M3) [above of = A3, xshift = 1cm, yshift=2mm, minimum
      width = 0mm] {$M_{\Pi, \adv}$}; \node (ML3) [left of = A3,
      xshift = 2cm, yshift = 0.3cm, draw, densely dotted, minimum
      height= 25 mm, minimum width = 3.6cm] {}; \draw[<->] (A3.east)
      -- (B3.west) node {}; \draw[->] (ina3.south) -- (A3.north) node
      {}; \draw[->] (inb3.south) -- (B3.north) node {}; \draw[->]
      (A3.south) -- (outa3.north) node {}; \draw[->] (B3.south) --
      (outb3.north) node {};

      \node (QC2) [right of = B3, xshift=.2cm] {$\qc$};

      \node (A4) [right of =QC2, xshift = .2cm, turn] {$\simulator$};
      \node (B4) [right of = A4, xshift = 1cm, turn] {$\tilde B$};
      \node (ina4) [above of=A4, yshift=2cm, minimum width = 0mm] {};
      \node (inb4) [above of=B4, yshift=2cm, minimum width = 0mm] {};
      \node (outa4) [below of=A4, yshift=-0.5cm, minimum width = 0mm] {};
      \node (outb4) [below of=B4, yshift=-0.5cm, minimum width = 0mm]
      {};
      \node (F2) [above of = A4, xshift = 1cm, yshift=5mm, smallturn] {{\color{blue}$\func$}};
      \node (M4) [above of = F2, yshift=-2mm, minimum width = 0mm] {$M_{\func, \simulator}$};
      \node (ML4) [left of = A4, xshift = 2cm, yshift = 0.8cm, draw, densely dotted, minimum height= 35 mm,
        minimum width = 3.6cm] {};
      \draw[<->] (A4.east) -- (B4.west) node {};
      \draw[->] (ina4.south) -- (A4.north) node {};
      \draw[->] (inb4.south) -- (B4.north) node {};
      \draw[->] (A4.south) -- (outa4.north) node {};
      \draw[->] (B4.south) -- (outb4.north) node {};
      \draw[<->] ([yshift=12mm]A4.east) -- ([xshift=-5mm]F2.south) node
      {};
      \draw[<->] ([yshift=12mm]B4.west) -- ([xshift=5mm]F2.south) node {};

    \end{tikzpicture}
  \end{center}
\vspace{-2ex}
\caption{Quantum stand-alone emulation between protocols (above) and
  realization of a functionality (below).
  } \label{fig:emuprot}
\end{figure}

 \begin{definition}[Statistically Quantum-Stand-Alone Emulation] Let
   $\Pi$ and $\Gamma$ be two poly-time protocols. We say $\Pi$
   \emph{statistically quantum-stand-alone} (\sqsa) emulates $\Gamma$,
   if for any QIM $\adv$ there exists an QIM $\simulator$ that runs in
   poly-time of that of $\adv$, such that
   $\mac_{\Pi, \adv} \dmm \mac_{\Gamma, \simulator}$.
 \label{def:qssa}
 \end{definition}

  We assume \emph{static} corruption only in this work, where the
  identities of corrupted parties are determined before protocol
  starts. The definitions above consider computationally bounded
  (poly-time) adversaries, including simulators. Occasionally, we will
  work with \emph{inefficient} simulators, which we formulate as
  unbounded simulation of corrupted party $P$.

  \begin{definition}[Unbounded Simulation of Corrupted $P$] Let $\Pi$
    and $\Gamma$ be two poly-time protocols. For any poly-time QIM
    $\adv$ corrupting party $P$, we say that $\Pi$ \cqsa-emulates
    $\Gamma$ against corrupted $P$ with unbounded simulation, if there
    exists a QIM $\calS$ possibly unbounded such that
    $ \mac_{\Pi,\adv} \qc \mac_{\Gamma,\simulator}$.
\label{def:uqcsa}
\end{definition}
\subsection{Modular Composition Theorem}
\label{sec:sacomposition}
It's shown that protocols satisfying the definitions of stand-alone
emulation admit a modular composition~\cite{HSS15}. Specifically, let
$\Pi$ be a protocol that uses another protocol $\Gamma$ as a
subroutine, and let $\Gamma'$ be a protocol that QSA emulates
$\Gamma$. We define the \emph{composed} protocol, denoted
$\cprot{\Pi}{\Gamma}{\Gamma'}$, to be the protocol in which each
invocation of $\Gamma$ is replaced by an invocation of $\Gamma'$. We
allow multiple calls to a subroutine and also using multiple
subroutines in a protocol $\Pi$. \textbf{However, quite importantly, we require that at any
point, only one subroutine call be in progress}. This is more
restrictive than the ``network'' setting, where many instances and
subroutines may be executed \emph{concurrently}.

In a \emph{hybrid} model, parties can make calls to an ideal-world
protocol $\tilde\Pi_{\funcG}$ of some functionality
$\funcG$\footnote{In contrast, we call it the \emph{plain model} if no
  such trusted set-ups are available.}. We call such a protocol a
\emph{$\funcG$-hybrid} protocol, and denote it $\Pi^{\funcG}$. The
execution of a hybrid-protocol in the presence of an adversary $\adv$
proceeds in the usual way. Assume that we have a protocol $\Gamma$
that realizes $\funcG$ and we have designed a $\funcG$-hybrid protocol
$\Pi^{\funcG}$ realizing another functionality $\func$. Then the
composition theorem allows us to treat sub-protocols as equivalent to
their ideal versions.

If the secure emulation involves unbounded simulation against a party,
the proof in~\cite{HSS15} can be extended to show that the composed
protocol also emulates  with unbounded simulation against
the corresponding corrupted party.

\begin{theorem}[Modular Composition] All of the following holds.
  \begin{itemize}
  \item Let $\Pi$, $\Gamma$ and $\Gamma'$ be two-party protocols such
    that $\Gamma'$ \cqsa-emulates $\Gamma$, then
    $\cprot{\Pi}{\Gamma}{\Gamma'}$ \cqsa~emulates
    $\Pi$. If $\Gamma'$ \cqsa~emulates $\Gamma$ against corrupted $P$
    with unbounded simulation, then   $\cprot{\Pi}{\Gamma}{\Gamma'}$
    \cqsa~emulates against corrupted $P$
    with unbounded simulation.

  \item Let $\func$ and $\funcG$ be poly-time functionalities. Let
    $\Pi^{\funcG}$ be a $\funcG$-hybrid protocol that \cqsa~realizes
    $\func$, and $\Gamma$ be a protocol that \cqsa~realizes $\funcG$,
    then $\cprot{\Pi}{\funcG}{\Gamma}$ \cqsa~realizes $\func$.  If
    $\Gamma$ \cqsa~realizes $\funcG$ against corrupted $P$ with
    unbounded simulation then $\cprot{\Pi}{\funcG}{\Gamma}$
    \cqsa~realizes $\func$ against corrupted $P$ with unbounded
    simulation.
  \end{itemize}
\label{thm:sacomposition}
\end{theorem}

\begin{figure}[ht!]
    \begin{center}
    \begin{tikzpicture}[scale=0.38,
        turn/.style={draw, minimum height=15mm, minimum width=10mm,
          fill = ChannelColor, text=ChannelTextColor},
        smallturn/.style={draw, minimum height=8mm, minimum width=8mm,
          fill = ChannelColor, text=ChannelTextColor},
        >=latex]
      \node (A) at (-8,0) [smallturn] {$\Gamma$};
      \node (inl) [above of=A, xshift=-.5cm, minimum width = 0mm] {};
      \node (inr) [above of=A, xshift=.5cm,  minimum width = 0mm] {};
      \node (outl) [below of=A, xshift=-.5cm, minimum width = 0mm]
      {};
      \node (outr) [below of=A, xshift=.5cm,  minimum
      width = 0mm] {};
      \node (omitt) [above of = A, minimum width
            = 0mm] {$\vdots$};
      \node (omitb) [below of = A, minimum width
      = 0mm] {$\vdots$};
      \node (P) [above of = A, yshift = 4mm, minimum width
      = 0mm] {$\Pi$};
      \node (ML) [left of = A, xshift=1cm, yshift = 2mm, , draw, densely dotted, minimum height= 32 mm,
        minimum width =1.8 cm] {};
      \draw[->] (inl.south) -- ([xshift=-5mm]A.north) node {};
      \draw[->] (inr.south) -- ([xshift=5mm]A.north) node {};
      \draw[->] ([xshift=-5mm]A.south) -- (outl.north) node {};
      \draw[->] ([xshift=5mm]A.south) -- (outr.north) node {};

      \node (App) [right of = A, xshift = .5cm] {$\approx$};

      \node (A1) [right of = App, xshift = .5cm, smallturn] {$\Gamma'$};
      \node (inl1) [above of=A1, xshift=-.5cm, minimum width = 0mm] {};
      \node (inr1) [above of=A1, xshift=.5cm,  minimum width = 0mm] {};
      \node (outl1) [below of=A1, xshift=-.5cm, minimum width = 0mm]
      {};
      \node (outr1) [below of=A1, xshift=.5cm,  minimum
      width = 0mm] {};
      \node (omitt) [above of = A1, minimum width
            = 0mm] {$\vdots$};
      \node (omitb) [below of = A1, minimum width
      = 0mm] {$\vdots$};
      \node (P1) [above of = A1, yshift = 4mm, minimum width
      = 0mm] {$\Pi^{\Gamma/{\Gamma'}}$};
      \node (ML1) [left of = A1, xshift=1cm, yshift = 2mm, , draw, densely dotted, minimum height= 32 mm,
        minimum width =1.8 cm] {};
      \draw[->] (inl1.south) -- ([xshift=-5mm]A1.north) node {};
      \draw[->] (inr1.south) -- ([xshift=5mm]A1.north) node {};
      \draw[->] ([xshift=-5mm]A1.south) -- (outl1.north) node {};
      \draw[->] ([xshift=5mm]A1.south) -- (outr1.north) node {};

      \node (L) [above of = A, yshift=1.5cm, smallturn] {$\Gamma$};
      \node (E) [right of = L, xshift=.5cm] {$\approx$};
      \node (R) [right of = E, xshift=.5cm, smallturn] {$\Gamma'$};
      \node (I) [below of = E] {$\Downarrow$};

      \node (A2) [right of = A1, xshift = 2.5cm, smallturn] {$\funcG$};
      \node (inl2) [above of=A2, xshift=-.5cm, minimum width = 0mm] {};
      \node (inr2) [above of=A2, xshift=.5cm,  minimum width = 0mm] {};
      \node (outl2) [below of=A2, xshift=-.5cm, minimum width = 0mm]
      {};
      \node (outr2) [below of=A2, xshift=.5cm,  minimum
      width = 0mm] {};
      \node (omitt2) [above of = A2, minimum width
            = 0mm] {$\vdots$};
      \node (omitb2) [below of = A2, minimum width
      = 0mm] {$\vdots$};
      \node (P2) [above of = A2, yshift = 4mm, minimum width
      = 0mm] {$\Pi$};
      \node (ML2) [left of = A2, xshift=1cm, yshift = 2mm, , draw, densely dotted, minimum height= 32 mm,
        minimum width =1.8 cm] {};
      \draw[->] (inl2.south) -- ([xshift=-5mm]A2.north) node {};
      \draw[->] (inr2.south) -- ([xshift=5mm]A2.north) node {};
      \draw[->] ([xshift=-5mm]A2.south) -- (outl2.north) node {};
      \draw[->] ([xshift=5mm]A2.south) -- (outr2.north) node {};

      \node (App2) [right of = A2, xshift = .5cm] {$\approx$};

      \node (A3) [right of = App2, xshift = .5cm, smallturn] {$\Gamma$};
      \node (inl3) [above of=A3, xshift=-.5cm, minimum width = 0mm] {};
      \node (inr3) [above of=A3, xshift=.5cm,  minimum width = 0mm] {};
      \node (outl3) [below of=A3, xshift=-.5cm, minimum width = 0mm]
      {};
      \node (outr3) [below of=A3, xshift=.5cm,  minimum
      width = 0mm] {};
      \node (omitt3) [above of = A3, minimum width
            = 0mm] {$\vdots$};
      \node (omitb3) [below of = A3, minimum width
      = 0mm] {$\vdots$};
      \node (P3) [above of = A3, yshift = 4mm, minimum width
      = 0mm] {$\Pi^{\funcG/{\Gamma}}$};
      \node (ML3) [left of = A3, xshift=1cm, yshift = 2mm, , draw, densely dotted, minimum height= 32 mm,
        minimum width =1.8 cm] {};
      \draw[->] (inl3.south) -- ([xshift=-5mm]A3.north) node {};
      \draw[->] (inr3.south) -- ([xshift=5mm]A3.north) node {};
      \draw[->] ([xshift=-5mm]A3.south) -- (outl3.north) node {};
      \draw[->] ([xshift=5mm]A3.south) -- (outr3.north) node {};

      \node (L2) [above of = A2, yshift=1.5cm, smallturn] {$\funcG$};
      \node (E2) [right of = L2, xshift=.5cm] {$\approx$};
      \node (R2) [right of = E2, xshift=.5cm, smallturn] {$\Gamma$};
      \node (I2) [below of = E2] {$\Downarrow$};

    \end{tikzpicture}
  \end{center}
\label{fig:modcomp}
\caption{Illustration of modular composition theorem: the general case (left) and in hybrid model (right).}
\end{figure}
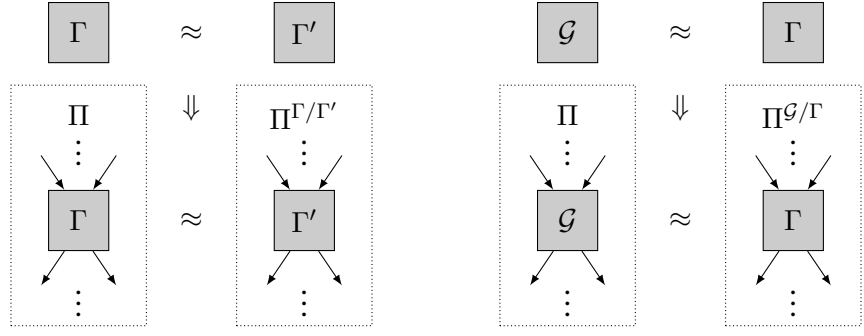

\section{Parallel OT  with Unbounded Simulation from \owf{}}
\label{sec:BBCS}
\label{sec:parallel-ot}

The goal of this section is to prove the following theorem.

\begin{theorem}
  Assuming the existence of \pqOWF{}, there exists a protocol $\Pi_{\pot}$ that $\cqsa$-emulates $\fpot$ with unbounded simulation against a malicious receiver.
\end{theorem}

We prove this theorem as follows. In \Cref{sec:bbcsot}, we review the protocol of \cite{C:BBCS91} that implies stand-alone-secure OT in $\fsocom$-hybrid model. Then, in \Cref{sec:parallel-composition}, we show how to build $\fpot$ from $\fsocom$. Finally in \Cref{sec:socom}, we construct $\fsocom$ with unbounded simulation against malicious sender. 

\subsection{Stand-Alone-secure OT in $\fsocom$-hybrid model}

\label{sec:bbcsot}\label{sec:qot}
In this section we present the quantum OT protocol assuming a selective opening-secure commitment scheme, that is, in the $\fsocom$~hybrid model. We would like to stress that the results in this section are not novel; they consist of a straightforward adaptation of previous results~\cite{C:BBCS91,DFL+09,EC:Unruh10} to our setting/language, and our goal in this presentation is to to provide a self-contained proof of its security.
We describe the protocol  $\qotbbcs$  in \Cref{sec:technical} and we have the following.

\begin{theorem}
\label{thm:ot}
$\qotbbcs$ $\cqsa$-realizes $\fot$ in the $\fsocom$ hybrid model.
\end{theorem}

Due to space restrictions and since it closely follows the proof in previous results, we defer the proof of \Cref{thm:ot} to \Cref{ap:qot}.

\subsection{Parallel Repetition for Protocols with Straight-Line Simulation}
\label{sec:parallel-composition}

\def\funcg{\mathcal{G}}

We show now that if $\pi$ implements $\func$ in the $\funcg$-hybrid model with an (efficient/unbounded) {\em straight-line} simulator, then a parallel repetition of $\pi$, denoted $\pi^{||}$ implements $\func^{||}$ in the $\funcg^{||}$-hybrid model with an (efficient/unbounded) simulator. As a corollary, we get that a parallel repetition of the $\func_{ot}$ protocol from the previous section is a secure implementation of parallel OT in the $\func_{\socom}$ hybrid model.

\begin{theorem}[Parallel Repetition]\label{thm:parallel}
  Let $\func$ and $\funcg$ be two-party functionalities and let $\pi$
  be a secure implementation of $\func$ in the $\funcg$-hybrid model
  with a {\em straight-line} simulator.  Then, $\pi^{||}$ is a secure
  implementation of $\func^{||}$ in the $\funcg^{||}$-hybrid model
  with straight-line simulation as well.
\end{theorem}

The proof of \Cref{thm:parallel} is deferred to \Cref{ap:parallel}.  We
immediately get the following by observing that parallel-$\fsocom$ is
exactly $\fsocom$.

\begin{corollary}
  The parallel repetition of any protocol that  $\cqsa$-realizes $\func_{ot}$ in the $\func_{\socom}$-hybrid model with a straight-line simulator achieves $\fpot$ in the $\func_{\socom}$-hybrid model.
\end{corollary}

\subsection{Implementing $\func_{\socom}$ with unbounded Simulation}
\label{sec:comm-unbounded}
\label{sec:socom}

In this section we provide an implementation of $\fsocom$ from Naor's commitment scheme and ZK protocols. Our protocol $\Pi_{\socom}$ is described in \Cref{fig:prot-ext-com} and we prove the following result.

\begin{theorem}\label{thm:socom}
Assuming the existence of \pqOWF{}, $\Pi_{\socom}$ $\cqsa$-realizes $\fsocom$. with unbounded simulation against malicious committer.
\end{theorem}
\vspace{-3ex}
\begin{figure}[!ht]
  \begin{mdframed}[style=figstyle,innerleftmargin=10pt,innerrightmargin=10pt]
{\sf Parties:} The committer $C$ and the receiver $R$.\\
{\sf Inputs:} $C$ gets $k$ $\ell$-bit strings $m_1$,...$m_k$ and $R$ gets a subset $I \subseteq [k]$ of messages to be decommited%

\begin{center}
\underline{\textsf{Commitment Phase}}
\end{center}
\begin{enumerate}
    \item $R$ sends $\rho$ for Naor's commitment scheme
    \item For $i \in [k]$, $C$ generates the commitments $c_i = \com_\rho(m_i,r_i)$, where $r_i$ is some private randomness.
    \item $C$ sends $c_1,...,c_k$ to $R$
\end{enumerate}

\begin{center}
\underline{\textsf{Decommitment Phase}}
\end{center}
\begin{enumerate}
   \item  $R$ sends $I$ to $C$
   \item $C$ sends $(m_i)_{i \in I}$ to $R$ and they run a ZK protocol to prove that there exists $\left((\tilde{m}_i)_{i \not\in I},(r_i)_{i \in [k]})\right)$ such that  $c_i = \com_{\rho}(\tilde{m}_i,r_i)$
\end{enumerate}
\vspace{-2ex}
\caption{Protocol for selective-opening commitment scheme $\Pi_{\socom}$.}\label{fig:prot-ext-com}
\end{mdframed}
\vspace{-3ex}
\end{figure}

We prove \Cref{thm:socom} by showing security against malicious committer with unbounded simulator in \Cref{lem:socom-sender} and security against malicious receiver in \Cref{lem:socom-receiver}.

\begin{lemma}\label{lem:socom-sender}
Assuming the existence of \pqOWF{},
$\Pi_{\socom}$ \cqsa-emulates
    $\fsocom$ against corrupted committer $\adv$ with unbounded simulation.
\end{lemma}
\begin{proof}
The unbounded simulator $\simulator$ works as follows:
\begin{enumerate}
\item In the commitment phase, $\simulator$ runs the honest protocol with $\adv$ and when receives the commitments $\hat{c}_1,...,\hat{c}_k$ from $\adv$ and $\simulator$ finds the messages $\hat{m}_1,...,\hat{m}_k$ by brute force. If there is a $\hat{c}_i$ that does not decommit to any message or decommits to more than one message $\simulator$ aborts. Finally, $\simulator$ inputs $\hat{m}_1,...,\hat{m}_k$ to $\func_{\socom}$
\item In the Decommitment phase,
   $\simulator$ receives $I$ from  $\func_{\socom}$, forwards it to $\adv$. $\simulator$ receives $(\tilde{m}_i)_{i \in I}$ from $\adv$ runs the honest verifier in the ZK protocol with $\adv$, and rejects iff the ZK rejects or if for any $i \in I$, $\hat m_i \ne \tilde m_i$.
\end{enumerate}

The proof follows the statistically-binding property of Naor's commitment scheme, so we can ignore commitments that open to more than one message, and by the ZK soundness property, which ensures that, up to negligible probability, if the commitments are not well-formed or if the sender tries to open then to a different value, both the simulator and the original receiver abort. 

Due to space restrictions, we leave the details to \Cref{ap:proof-sender-socom}
\end{proof}

We now show security against malicious receiver.

\begin{lemma}\label{lem:socom-receiver}
Assuming the existence of \pqOWF{},
$\Pi_{\socom}$ \cqsa-realizes
    $\fsocom$ against corrupted receiver $\adv$.
\end{lemma}
\begin{proof}
The simulator $\simulator$ works as follows:
\begin{enumerate}
\item In the commitment phase,  $\simulator$ sends $c_i = \com_\rho(0,r_i)$ to $\adv$
\item In the decommitment phase, $\simulator$ receives $I$ from  $\adv$, uses it as input of $\func_{\socom}$. $\simulator$ receives back the messages $(m_i)_{i \in I}$, sends them to $\adv$ and runs the ZK simulator of the proof that $(c_i)_{i \in I}$ open to $(m_i)_{i \in I}$ and that $(c_i)_{i \not\in I}$ are valid commitments.
\end{enumerate}

The fact that $ M_{\Pi_{\socom},\adv} \qc M_{\fsocom, \simulator}$
follows from the computational zero-knowledge of the protocol and the
computatinally-hiding property of Naor's commitment scheme. We defer
the details of the proof to \Cref{ap:proof-receiver-socom}

\end{proof}

\fsnotem{I noticed several protocols are described in FZK
  hybrid. Strictly speaking, we cannot apply composition later since
  we don't have fully-simulatable ZK protocols assuming OWFs alone. I
  suppose we can just say "run ZK protocols" instead. Incidentally,
  Watrous proved quantum ZK for NP using a non-interactive commitment,
  and his theorem didn't claim quantum-secure ZK from OWFs. I don't
  see any issue plugging in Naor's commitment though.}

\section{Extractable Commitment from Unbounded Simulation OT}
\label{sec:ecom}

In this section, we construct an extractable commitment scheme using the
unbounded simulation OT from section~\ref{sec:BBCS}.
We do this in two steps. First, we define a new primitive, namely {\em verifiable}
conditional disclosure of secrets (vCDS) in section~\ref{sec:cdsdef}, and we
construct a (unbounded simulation)
vCDS protocol in section~\ref{sec:cdsprot} from the unbounded simulation OT. We then show how to use
vCDS to construct an extractable commitment protocol that implements
$\fsocom$ with efficient simulators in section~\ref{sec:cds-to-extcom}.

\subsection{Verifiable Conditional Disclosure of Secrets (vCDS)}
\label{sec:cdsdef}

\agnotem{Is it necessary to define $\func_{CDS}$? I am not sure that we are using it in the later proof, since we need to use concrete protocol and we don't use the ideal functionality.}

\rachelm{I think we do implement $\func_{CDS}$ with unbounded
  simulation for both malicious sender and receiver. In the case of
  malicious receiver, we can simulate OT to get all the input it is
  using in each copy of OT. The simulator can test if any of them is a
  valid witness of $x$. If so, it sends that witness to the ideal
  functionality to get $m$ and then simulate the garbled circuits
  accordingly. So it seems we are fine, am I missing anything?}

\rachelm{Defined the above only w.r.t.\ 2 party and only has sid.}
\vnotem{need to ref for CDS. Who?}

We define the primitive of (verifiable) conditional disclosure of secrets.
Conditional disclosure of secrets~\cite{GIKM98} (CDS) for an $\NP$-language $\Lang$
is a two-party protocol where a sender (denoted $\cdsS$) and a receiver
(denoted $\cdsR$) have a common input $x$, the sender has a message $\mu$,
and the receiver (purportedly) has a witness $w$ for the $\NP$-relation $R_{\Lang}$.
At the end of the protocol, $\cdsR$ gets $\mu$ if $R_{\Lang}(x,w)=1$
and $\bot$ otherwise, and the sender gets nothing. In a sense, this can be
viewed as a {\em conditional} version of oblivious transfer, or as an interactive
version of witness encryption.

The CDS functionality is defined in Figure~\ref{fig:fcds}. We will
construct a protocol $\Pi = \inter{\cdsS, \cdsR}$ that securely
realizes the CDS functionality in the quantum stand-alone model. We
will consider protocols with either efficient or unbounded simulators.

\begin{figure}[!ht]
  \begin{mdframed}[style=figstyle,innerleftmargin=10pt,innerrightmargin=10pt]
  \begin{center}
    \textbf{The Conditional Disclosure of Secret (CDS) Functionality}
    $\func_{CDS}$ for an NP language $\mathcal{L}$.
  \end{center}
  \
  \\
  \textsf{Security Parameter:} $\lambda$.  \\
    \textsf{Parties:} Sender $S$ and Receiver $R$, adversary $\adv$.

  \begin{description}
  \item[Sender Query:] $\func_{CDS}$ receives
    $(\text{Send},sid,(x,\mu))$ from $S$, where
    $x \in \Lang \cap \bits^{n_1(\secpar)}$ and
    $m \in\bits^{n_2(\secpar)}$ for polynomials $n_1$ and $n_2$,
    records $(sid,(x, \mu))$
    and sends $(\text{Input},sid,x)$ to $R$ and $\mathcal{A}$.

    $\func_{CDS}$ ignores further send messages from $S$ with $sid$.
  \item[Receiver Query:] $\func_{CDS}$ receives
    $(\text{Witness},sid, w)$ from party $R$, where
    $w \in\bits^{m(\secpar)}$ for a polynomial $m$.  $\func_{CDS}$
    ignores the message if no $(sid,\star)$ was recorded. Otherwise
    $\func_{CDS}$ sends $(\text{Open},sid, x, \mu')$ to $R$ where
    $$\mu' = \left\{ \begin{array}{cc}
    \mu & \mbox{if $\Rel_\Lang(x, w) = 1$} \\
    \bot & \mbox{if $\Rel_\Lang(x, w) = 0$} \end{array} \right.$$
    $\func_{CDS}$ sends
    $(\text{Open},sid, x)$ to $\adv$ and ignores further
    messages from $R$ with $sid$.
\end{description}
\end{mdframed}
\caption{The Conditional Disclosure of Secrets (CDS) Functionality}
\label{fig:fcds}
\vspace{-3ex}
\end{figure}

\paragraph{Verifiability.}
We will, in addition, also require the CDS protocol to be {\em verifiable}.
Downstream, when constructing our extractable commitment protocol in Section~\ref{sec:cds-to-extcom},
we want to be able to prove consistency of the transcript of a
CDS sub-protocol. It is not a-priori
clear how to do this since the CDS
protocol we construct will either live in the OT-hybrid model, in which case
the OT input is {\em not} contained in the protocol transcript and is
unconstrained by it; or it uses quantum communication, in which case,
again consistency cannot be expressed as an NP-statement.

\rachelm{it would be more general to define verifiability w.r.t.\ SFE,
  but that would require defining the ideal functionality F
  for computing a function f, which we did not define.}

\begin{definition}[Verifiability]
  \label{def:ver}
  Let $\Lang$ be an $\NP$ language, and $\Pi = \inter{\cdsS, \cdsR}$ be a
  CDS protocol between a sender $\cdsS$ and a receiver $\cdsR$. $\Pi$ is
  {\em verifiable} (w.r.t.\ $\cdsS$) if there is a polynomial time
  classical algorithm $\Ver$, such that, the following properties are
  true:
  \begin{description}
  \item[Correctness:] For every
    $(x,\mu)$ and every $w$, $\cdsS(x,\mu)$ after interacting with $\cdsR(w)$,
    outputs on a special output tape a proof $\pi$, such that,
    $\Ver(\tau,x,\mu,\pi) = 1$ where $\tau$ is the transcript of classical
    messages exchanged in the interaction.

  \item[Binding:] For every $\secpar \in \Nat$, every
    (potentially unbounded) adversary
    $\adv = \set{\adv_\secpar}_{\secpar\in\Nat}$, every sequence of
    witnesses $\set{w_\secpar}_\secpar$,
    the probability that $\adv_\secpar$ wins in the following experiment is
    negligible.

  \begin{itemize}
  \item $\adv_\secpar$ after interacting with $\cdsR(1^\secpar, w)$,
    outputs $(x,\mu,\pi)$. Let $\tau$ be the transcript of classical
    messages exchanged in the interaction.

  \item $\adv_\secpar$ wins if (a) $\Ver(\tau, x, \mu, \pi) = 1$, (b) $\cdsR$ did
    not abort, and (c) $\cdsR$ outputs $\mu'$ inconsistent with inputs $(x,\mu)$ and
    $w$, that is,
    \begin{align*}
      \mu' \ne
      \begin{cases}
        \mu & \text{ if } \Rel_\Lang(x, w) = 1  \\
        \bot & \text{ otherwise }
      \end{cases}
    \end{align*}
  \end{itemize}
  \end{description}
\end{definition}

\begin{definition}[Verifiable CDS]
  Let $\Lang$ be an $\NP$ language, and $\Pi = \inter{\cdsS, \cdsR}$ be a
  protocol between a sender $\cdsS$ and a receiver $\cdsR$. $\Pi$ is a
  {\em verifiable CDS} protocol if (a) it $\cqsa$-emulates $\fcds$ with
  an efficient simulator; and (b) it is verifiable according to Definition~\ref{def:ver}.
\end{definition}

\subsection{CDS Protocol from Unbounded Simulation OT}
\label{sec:cdsprot}

\begin{theorem}\label{thm:vercds}
  Assume the existence of \pqOWF{}.  For
  every $\NP$ language $\Lang$, there is a verifiable CDS protocol
  $\Pi = \inter{\cdsS, \cdsR}$ that \cqsa-emulates $\fcds$ for
  $\Lang$ in the $\fpot$ hybrid model.
\end{theorem}

\begin{corollary}
  Assume the existence of \pqOWF{}, and a
  protocol that \cqsa-emulates $\fpot$ with unbounded
  simulation. Then, for every $\NP$ language $\Lang$, there is a
  verifiable CDS protocol $\Pi = \inter{\cdsS, \cdsR}$ that
  \cqsa-emulates $\fcds$ for $\Lang$ with unbounded simulation.
\end{corollary}

 \paragraph{Proof of Theorem~\ref{thm:vercds}.}
 The verifiable CDS protocol is described in Figure~\ref{fig:cdsprot}.
 The protocol uses Naor's classical statistically binding commitment
   protocol, Yao's garbled circuits, and post-quantum zero knowledge proofs, all
   of which can be implemented from \pqOWF{}. For a more
   detailed description of these ingredients, see Section~\ref{sec:cryptoprelims}.

 In lemma~\ref{lem:rcvrsim}, we show that the protocol has an efficient simulator for a corrupted receiver,
 and in lemma~\ref{lem:sendsim}, an efficient simulator for a corrupted sender (both in the OT hybrid model).
 Lemma~\ref{lem:verifiability} shows that the protocol is verifiable.
 \qed

\begin{figure}[htpb]
{\small
 \begin{mdframed}[style=figstyle,innerleftmargin=10pt,innerrightmargin=10pt]
 \textsf{Parties:} The sender $\cdsS$ and the receiver $\cdsR$. \hfill \textsf{Inputs:} $\cdsS$ has input $(x,\mu)$ and $\cdsR$ has input $w$. 
  \begin{enumerate}
  \item \textbf{\textcolor{blue}{Preamble}}:
  $\cdsR$ sends a random string $\rho$ as the first message of
    Naor's commitment scheme to $\cdsS$ and $\cdsS$ sends $x$ to $\cdsR$

   \item \textbf{Compute Garbled Circuits:}
   $\cdsS$ generates $2\secpar$ garbled circuits, for the circuit
     computing
     \begin{align*}
       G_{x,\mu}(w) =  \mu' = \begin{cases}
          \mu & \text{ if } \Rel_\Lang(x, w) = 1  \\
          \bot & \text{ otherwise }
        \end{cases}
     \end{align*}
     That is, for every $i \in [2\secpar]$,
     $$(\hat G^i,
     \smallset{\ell^{i,j}_b}_{j \in [m], b \in \bits}) = \Garb(G_{x,\mu};
     \gamma_i)$$
    where $m$ is the length of the witness, $\hat G^i$ are
     the garbled circuits, and $\ell$'s are its associated labels.

  \item \textbf{Cut-and-Choose:}
     $\cdsR$ samples a random subset $\Lambda \subseteq [2\secpar]$, by
       including each $i \in [2\secpar]$ with probability 1/2. For every
       $i \in [2\secpar]$, set
       \begin{align*}
        \sigma^i =
          \begin{cases}
            s^i \gets \bits^m & i \in \Lambda \\
            w & i \not\in \Lambda
          \end{cases}
       \end{align*}

   \item \textbf{\textcolor{blue}{OT}:}
   For every $i \in [2\secpar], j\in [m], b \in \bits$, $\cdsS$ samples
   $r^{i,j}_b$, the random coins for
   committing to the labels $\ell^{i,j}_b$ via Naor's commitment scheme.

   $\cdsS$ and $\cdsR$ invokes $\fpot$ for $2\secpar \times m$
     parallel OT, where the $(i,j)$'th OT for
     $i \in [2\secpar], j \in [m]$ has sender's input strings
     $(\ell^{i,j}_0, r^{i,j}_0)$ and $(\ell^{i,j}_1, r^{i,j}_1)$, and
     receiver's choice bit $\sigma^{i,j}$ (which is the $j$-th bit of $\sigma^{i}$) and $\cdsR$ receives
     $(\tilde \ell^{i,j}, \tilde r^{i,j})$.

     We refer to the OTs with index $(i, \star)$ as the $i$'th batch.
     as they transfer labels of the $i$'th garbled circuit $\hat G_i$.

  \item \textbf{\textcolor{blue}{Send Garbled Circuits and Commitments to the Labels and $\mu$}:}
         $\cdsS$ samples $r^*$ and computes $c^* = \com_{\rho}(\mu; r^*)$ and
         $c^{i,j}_b = \com_{\rho}(\ell^{i,j}_b; r^{i,j}_b)$.

         Send $\smallset{\hat G^i}_{i \in [2\secpar]}$ and
         $(c^*, \smallset{c^{i,j}_b}_{i\in [2\secpar], j \in [m], b \in \bits})$ to the receiver $\cdsR$.

   \item \textbf{\textcolor{blue}{Proof of Consistency}:}
   $\cdsS$ proves via ZK protocol that (a) $c^*$ is a valid commitment to
     $\mu$, (b) every $\hat G^i$ is a valid garbling of $G_{x,\mu}$ with labels
     $\smallset{\ell^{i,j}_b}_{j \in [m], b \in \bits}$, and (c) $c^{i,j}_b$ is a
     valid commitment to $\ell^{i,j}_b$.

   \item \textbf{Checks:} $\cdsR$ performs the following checks:
     \begin{itemize}
     \item If the ZK proof in the previous step is not accepting, $\cdsR$
       aborts.

     \item If there is $i \in \Lambda$ and $j \in [m]$, such that,
       $c^{i,j}_{\sigma^{i,j}} \ne \com_{\rho}(\tilde \ell^{i,j}, \tilde
       r^{i,j})$, $\cdsR$ aborts and outputs $\err_1$.

     \item If for every $i \not\in \Lambda$, there exists $j \in [m]$, such
       that, $c^{i,j}_{\sigma^{i,j}} \ne \com_{\rho}(\tilde \ell^{i,j},
       \tilde r^{i,j})$, $\cdsR$ aborts and outputs $\err_2$.
     \end{itemize}

   \item \textbf{Output:} If $\cdsR$ does not abort, there must exist $i \not \in \Lambda$
     such that, for all $j \in [m]$,
     $c^{i,j}_{\sigma^{i,j}} = \com_{\rho}(\tilde \ell^{i,j}, \tilde
     r^{i,j})$. Evaluate the $i$'th garbled circuit $\hat G^i$ to get
      $\mu' = \GEval(\hat G^i, \smallset{\tilde \ell^{i,j}}_{j\in[m]})$,
     and output $x',\mu'$. %

 \end{enumerate}

 \caption{\small The verifiable CDS Scheme in $\fpot$-hybrid model. The steps in color involve communication while the others only involve local computation. }\label{fig:cds-scheme}\label{fig:cdsprot}
 \end{mdframed}
 }
 \end{figure}

\begin{lemma}\label{lem:rcvrsim}
  There is an efficient simulator against a malicious receiver.
\end{lemma}

\begin{proof}
  The simulator $\simulator$ interacts with $\cdsR^*$, receives a string $\rho$ from
  $\cdsR^*$ in Step $1$, and intercepts
  the OT queries $(\sigma^{1},\ldots,\sigma^{2\secpar})$ in Step $4$.
  \begin{itemize}
    \item \textbf{Case $1$. $R_{\Lang}(x,\sigma^{i})=1$ for some $i$.}
    Send $(\text{Witness},sid, \sigma^{i})$ to the CDS functionality and receive
    $\mu$. Simulate the rest of the protocol honestly using the CDS sender input
    $(x,\mu)$.
    \item  \textbf{Case $2$. $R_{\Lang}(x,\sigma^{i})=0$ for all $i$.}
    Simulate the rest of the protocol honestly using the CDS sender input
    $(x,0)$.
  \end{itemize}

  We now show, through a sequence of hybrids, that this simulator produces a view that is computationally
  indistinguishable from
  that in the real execution of $\cdsS(x,\mu)$ with $\cdsR^*$.

  \medskip\noindent
  \textit{Hybrid $0$.} This corresponds to the real execution of the protocol where the
  sender has input $(x,m)$. The view of $\cdsR^*$ consists of
  $$ \bigg[ \rho,
           \smallset{\hat G^i, \tilde \ell^{i,j}, \tilde r^{i,j}, c^{i,j}_b}_{i\in[2\secpar],j\in[m],b\in\{0,1\}},
           c^*,
           \tau_{\mathsf{ZK}}    \bigg]$$
  where $\rho$ is the message sent by $\cdsR^*$ in Step $1$, the strings $\tilde \ell^{i,j}$ and $\tilde r^{i,j}$
  are received by $\cdsR^*$ from the OT functionality in Step $4$, the garbled circuits $\hat G^i$ and
  the commitments $c^{i,j}_b$ and $c^*$ in Step $5$, and $\tau_{\mathsf{ZK}}$ is the transcript of the ZK
  protocol between $\cdsS$ and $\cdsR^*$ in Step $6$. (See the protocol in Figure~\ref{fig:cdsprot}).

  \medskip \noindent
  \textit{Hybrid $1$.}
  This is identical to hybrid $0$ except that we run the simulator to intercept the OT queries
  $(\sigma^{1},\ldots,\sigma^{2\secpar})$ of $\cdsR^*$.
  The rest of the execution remains the same. Of course, the transcript produced is identical to that in
  hybrid $0$.

  \medskip \noindent
  \textit{Hybrid $2$.}
  In this hybrid, we replace the transcript $\tau_{\mathsf{ZK}}$ of the zero-knowledge
  protocol with a simulated transcript. This is indistinguishable from hybrid $1$ by
  (post-quantum) computational zero-knowledge. Note that generating this hybrid does not require
  us to use the randomness underlying the commitments $c^{i,j}_{1-\sigma^{i,j}}$ and $c^*$.
  (The randomness underlying $c^{i,j}_{\sigma^{i,j}}$ {\em are} revealed as part of the
  OT responses to $\cdsR^*$.)

  \medskip \noindent
  \textit{Hybrid $3$.}
  In this hybrid, we replace half the commitments, namely $c^{i,j}_{1-\sigma^{i,j}}$, as well as
  $c^*$ with commitments of $0$. This is indistinguishable from hybrid $2$ by
  (post-quantum) computational hiding of Naor commitments.

  \medskip \noindent
  \textit{Hybrid $4$.}
  In this hybrid, we proceed as follows.
  If the simulator is in case $1$, that is $R_{\Lang}(x,\sigma^i) = 1$ for some $i$,
  proceed as in hybrid $3$ with no change. On the other hand,
  if the simulator is in case $2$,  that is $R_{\Lang}(x,\sigma^i) = 0$ for all $i$,
  replace the garbled circuits with simulated garbled circuits that always output $\bot$ and
  let the commitments $c^{i,j}_{\sigma^{i,j}}$ be commitments of the simulated labels.
  This is indistinguishable from hybrid $3$ where the garbled circuits are an honest
  garbling of $G_{x,\mu}$ because of the fact that all the garbled evaluations output $\bot$ in hybrid $3$,
  and because of the post-quantum security of the
  garbling scheme.

  Hybrids $5$--$7$ undo the effects of hybrids $2$--$4$ in reverse.

  \medskip \noindent
  \textit{Hybrid $5$.}
  In this hybrid, we replace the simulated garbled circuit with the real garbled circuit for the circuit $G_{x,0}$.
  This is indistinguishable from hybrid $4$ because of the fact that all the garbled evaluations output $\bot$ in this hybrid,
  and because of the post-quantum security of the
  garbling scheme.

  \medskip \noindent
  \textit{Hybrid $6$.}
  In this hybrid, we let all commitments be to the correct labels and messages. This is indistinguishable from hybrid $5$ by
  (post-quantum) computational hiding of Naor commitments.

  \medskip \noindent
  \textit{Hybrid $7$.}
  In this hybrid, we replace the simulated ZK transcript
  with the real ZK protocol transcript. This is indistinguishable from hybrid $7$ by
  (post-quantum) computational zero-knowledge.

  This final hybrid matches exactly the simulator. This finishes the
  proof.
\end{proof}
\begin{lemma}\label{lem:sendsim}
  There is an inefficient statistical simulator against a malicious sender.
\end{lemma}
\begin{proof}
  The simulator $\simulator$ interacts with $\cdsS^*$ as follows:
  \begin{itemize}
    \item Sending a string $\rho$ to $\cdsS^*$ in Step $1$, as in the protocol;
    \item {\em Intercept} the OT messages $(\ell^{i,j}_0,r^{i,j}_0)$ and $(\ell^{i,j}_1,r^{i,j}_1)$
    from $\cdsS^*$ in Step $4$.
    \item Run the rest of the protocol as an honest receiver $\cdsR$ would. If the
  verifier of the ZK proof rejects, send $(x,\bot)$ to the ideal functionality and halt.
    \item  Label the $i$-th
  garbled instance {\em bad} if for some $j \in [m]$ and $b\in \{0,1\}$,
  the label $\ell^{i,j}_b$  together with the decommitment $r^{i,j}_b$ is not
  consistent with the commitment $c^{i,j}_b$.
  \begin{itemize}
    \item If more than $\secpar$ garbled instances are bad, send $(x,\bot)$ to the ideal functionality and halt.
    \item If not, extract $\mu$ from $c^*$ using unbounded time, and send $(x,\mu)$ to the ideal functionality and halt.
  \end{itemize}
  \end{itemize}
  The transcript generated by $\simulator$ is identical to the one generated in
  the real world where $\cdsR$ on input $w$ interacts with $\cdsS^*$. It remains to
  analyze the output distribution of $\cdsR$ in the simulation vis-a-vis the real world.

  We split the analysis into two cases.
  \begin{enumerate}
    \item If more than $\secpar$ garbled instances are bad, $\simulator$ will send
    $(x,\bot)$ to the ideal functionality; on the other hand, the receiver will also output $\bot$ except
    with probability at most $2^{-\secpar/2}$, since the expected number of bad garbled circuits in $\Lambda$ is
    at least $\secpar/2$ and the probability that all of them check out is $(1/2)^{\secpar/2}$.
    \item If
    fewer than $\secpar$ garbled instances are bad, $\simulator$ will extract $\mu$ and send
    $(x,\mu)$ to the ideal functionality; on the other hand, in the real world, (1) at least
    one garbled instance is not bad; (2) since the ZK proof checked out,
    we know that all garbled circuits contain the same circuit $G_{x,\mu}$ with the correct
    labels committed in $c^{i,j}_b$; and (3) if all the commitment checks pass, the output of the garbled
    evaluation must be $\mu$.
  \end{enumerate}
 Thus, we have that the output distributions of the receiver are negligibly close between the
 simulation and the real world, finishing up the proof.
\end{proof}

 \begin{lemma}\label{lem:verifiability}
    The protocol is verifiable.
 \end{lemma}
 \begin{proof}
   We first construct a verification algorithm $\Ver$.
   \begin{itemize}
    \item The transcript $\tau$ of classical messages consists of $\rho,
      x, \smallset{\hat G^i}_{i \in [2\secpar]}, c^*,
      \smallset{c^{i,j}_b}_{i \in [2\secpar], j\in[m], b \in \bits}$.
      \fsnotem{what is $r_S$? Sender's randomness? how to define when
      sender is quantum?} \vnotem{we are talking here of the honest sender who is classical}

    \item At the end of the protocol, $\cdsS$ outputs $(x, \mu, r^*)$ on its special output tape.

    \item The verification algorithm $\Ver(\tau, x, \mu', r') = 1$ iff $c^*=\com_{\rho}(\mu'; r')$.
   \end{itemize}

 We first claim that for honest $\cdsS$ and $\cdsR$ with $(x,w) \in \mathcal{R}_{\mathcal{L}}$, we have that $\Ver(\tau,x,\mu,r) = 1$.
 Since all parties in the protocol are honest the input $x$ in $\tau$ is the same as the one output by $\cdsS$ and we have that $c^*$ is the commitment to the honest message using the correct randomness, so $\Ver$ outputs $1$.

 To show binding, assume that the verification passes and the receiver does not abort. Then, we know that there is at least
 one $i \notin \Lambda$ such that the $i$-th garbled circuit+input pair is correct and the circuit is the garbling of $G_{x,\mu}$. The verifier will evaluate the circuit on input $w$ and obtain either $\bot$ when $R_{\Lang}(x,w) = 0$ or $\mu$ when $R_{\Lang}(x,w)=1$, exactly as required.
 \end{proof}

\subsection{Extractable Commitment from CDS}
\label{sec:cds-to-extcom}

\begin{theorem}\label{thm:ecom}
  Assume the existence of \pqOWF{}.
  There is a commitment protocol $\inter{C, R}$ that
  \cqsa-emulates $\fsocom$ with efficient simulators.%
 \end{theorem}
\begin{proof}
The construction of our extractable commitment scheme is given in Figure~\ref{fig:ext-com-scheme}.
The protocol uses Naor's classical statistically binding commitment
  protocol and a verifiable CDS protocol $\Pi = \inter{\cdsS, \cdsR}$ that
    \cqsa-emulates
    $\fcds$ (with unbounded simulation) for $\Lang_\com$, the
    language consisting of all Naor's commtiments $(\rho,c)$ to a bit $b$:
    \begin{align*}
      \Rel_{\Lang_\com} ((\rho,c,b), r) = 1 \text{ iff } c = \com_{\rho}(b;r) ~.
    \end{align*}
  For a more
  detailed description of these ingredients, see Section~\ref{sec:cryptoprelims} and \ref{sec:cdsprot}.

In~\Cref{lem:commsim} (resp.~\Cref{lem:commrcvrsim}, we show that the protocol has an efficient simulator for a corrupted sender (resp. receiver).  %
\end{proof}
\vspace{-3ex}
\begin{figure}[!h]
\begin{mdframed}[style=figstyle,innerleftmargin=10pt,innerrightmargin=10pt]
\textsf{Parties:} The committer $C$ and the receiver $R$.\\
\textsf{Inputs:} $C$ gets a message vector $\vec{\mu}=(\mu_1,\ldots,\mu_{\ell(n)})$ and $R$ gets $1^n$.
\begin{center}
\underline{\textsf{Commitment Phase}}
\end{center}
 \begin{enumerate}
 \item \textbf{Preamble.} $C$ sends a random string $\rho$ to $R$, and $R$ sends a random string $\rho^*$ to $C$, as the first message of the Naor commitment scheme.
 \item \textbf{Set up a Trapdoor Statement.}
 \begin{itemize}
   \item
   $R$ sends a Naor commitment $c = \com_{\rho}(0;r)$.
   \item
   $R$ proves to $C$ using a ZK protocol that $c$ is a commitment to $0$, that is,
   $((c,\rho,0), r) \in \mathcal{R}_{\mathcal{L}_{\mathsf{com}}}$.
   If the ZK verifier rejects, $C$ aborts.
 \end{itemize}

 \item \textbf{CDS.}
 $C$ and $R$ run the CDS protocol $\inter{\cdsS, \cdsR}$ for the language $\mathcal{L}_{\com}$ where $C$ acts as
 $\cdsS$ with input $x = (c,\rho, 1)$ and message $\vec{\mu}$, and $R$ acts as $\cdsR$ with input $0$. %

 $C$ aborts if $\cdsS$ aborts, else $C$ obtains the protocol transcript $\tau$ and $\cdsS$'s proof $\pi$.
 $R$ aborts if $\cdsR$ aborts, or if $\cdsR$
 outputs $(x', \vec{\mu}')$ but $x'\ne (\rho, c, 1)$.

 \item \textbf{Commit and Prove Consistency.}
 \begin{itemize}
   \item $C$ sends a Naor commitment $c^* = \com_{\rho^*}(\vec{\mu};r^*)$.
   \item $C$ proves to $R$ using a ZK protocol there exists a $\vec{\mu}$ such that $(x=(\rho, c, 1), \vec{\mu})$ is the input that $C$ used in the CDS protocol and $\vec{\mu}$ is committed in $c^*$, that is:
   $$\Ver(\tau, x, \vec{\mu}, \pi) = 1 \mbox{ and } c^* = \com_{\rho^*}(\vec{\mu}, r^*)$$
 \end{itemize}

 \item $R$ accepts this commitment if the ZK proof is accepting.
\end{enumerate}
\begin{center}
\underline{\textsf{Decommitment Phase}}
\end{center}

\begin{enumerate}
\item $R$ sends $I \subseteq [\ell]$.
 \item $C$ sends $\vec{\mu}|_I$ and proves via a ZK protocol that
   $c^*|_I$ commits to $\vec{\mu}|_I$.

 \item $R$ accepts this decommitment if the ZK proof is accepting.

 \end{enumerate}
\vspace{-3ex}
\caption{Extractable Selective-Opening-Secure Commitment Scheme}\label{fig:ext-com-scheme}
\end{mdframed}
\end{figure}

\begin{lemma}\label{lem:commsim}
There is an efficient simulator against a malicious sender. %
\end{lemma}
\begin{proof}
  The simulator $\simulator$ against a malicious committer $C^*$ works as follows.
  \begin{enumerate}
  \item In step $1$, proceed as an honest receiver would.
  \item In step $2$, send a Naor commitment $c = \com_{\rho}(1;r)$ (instead of $0$)
  and simulate the ZK proof.
  \item In step $3$, run the honest CDS protocol with $r$ as witness, gets $\vec{\mu}$ and sends it
  to the ideal functionality $\fsocom$.
  \item Run the rest of the protocol as an honest receiver would.
    \end{enumerate}

We now show, through a sequence of hybrids, that this simulator produces a joint distribution of a view
of $C^*$ together with an output of $R$ that is computationally
        indistinguishable from that in the real execution of $C^*$ with $R$.
In order to show this we consider the following sequence of  hybrids.

\medskip\noindent
\textit{Hybrid 0.} This corresponds to the protocol  $\qecomhzero$, where $\simulator_0$ sits between $C^*$ and the honest receiver in the real protocol and just forwards their messages. It follows trivially that
$M_{\Pi_{\ecom},C^*} \qc
    M_{\qecomhzero, \simulator_0}$.

\medskip\noindent
\textit{Hybrid 1.}
 $\simulator_1$ interacts with $C^*$ following the protocol $\qecomhone$, which is the same as $\qecomhzero$ except that $\simulator_1$ uses the ZK simulator instead of the the proof that
 $((c,\rho,0), r) \in \mathcal{R}_{\mathcal{L}_{\mathsf{com}}}$.
From the computational zero-knowledge property of the protocol, we have that
$    M_{\qecomhzero, \simulator_0} \qc
    M_{\qecomhone, \simulator_1}$.

\medskip\noindent
\textit{Hybrid 2.}
$\simulator_2$ interacts with $C^*$ following the protocol $\qecomhtwo$, which is the same as $\qecomhone$ except that $\simulator_2$ sends $c' = \com_{\rho}(1;r)$ instead of the (honest) commitment of $0$. When $\simulator_2$ simulates $\fzk$, she still sends a message that $c'$ is a valid input. It follows from computationally hiding property of Naor's commitment scheme that $M_{\qecomhone, \simulator_1} \qc
    M_{\qecomhtwo, \simulator_2}$.

\medskip\noindent
\textit{Hybrid 3.}
$\simulator_3$ interacts with $C^*$ following the protocol $\qecomhthree$, which is the same as $\qecomhtwo$ except that $\simulator_3$ now uses the private randomness $r$ as a witness that $c'$ is a commitment of $1$.

Since our protocol realizes $\func_{CDS}$, $\cdsS^*$ (controlled by $C^*$) does not behave differently depending on the input of $\cdsR$, so the probability of abort in step $3$ does not change. Notice also that  $\mathsf{Ver}(\tau,x,\vec{\mu},\pi)$ is independent of $\cdsR$'s message, so the acceptance probability of the ZK proof does not change either.

Then, if the ZK proof leads to acceptance, by the soundness of the protocol, we know that $\mathsf{Ver}(\tau,x,\vec{\mu},\pi) = 1$ and
by the binding of the commitment $c^*$, such a $\vec{\mu}$ is uniquely determined.

Finally, by the verifiability of the CDS protocol, we know that the receiver either aborts or outputs the specified $\vec{\mu}$. Thus, the outputs of the receiver $R$ in the simulated execution and the real execution must be the same in this case.

\end{proof}

\begin{lemma}\label{lem:commrcvrsim}
  There is an efficient simulator against a malicious receiver.
\end{lemma}
\begin{proof}
  The simulator $\simulator$ against a malicious receiver $R^*$ proceeds as follows.
  \begin{itemize}
    \item In steps $1$ and $2$, proceed as an honest sender would.
    \item In step $3$, run the CDS protocol using a message vector $\vec{\mu} = \vec{0}$ of all zeroes.
    \item In step $4$, commit to the all-$0$ vector and produce a simulated ZK proof.
    \item During decommitment, send $I \subseteq [\ell]$ to the ideal functionality and receive $\vec{\mu}|_I$.
    Send $\vec{\mu}|_I$ to $R^*$, and simulate the ZK proof.
    \end{itemize}

  We now show, through a sequence of hybrids, that this simulator produces a view that is computationally
    indistinguishable from that in the real execution of $C(\vec{\mu})$ with $R^*$.

    \medskip\noindent
    \textit{Hybrid $0$.} This corresponds to the protocol  $\qecomhzero$, where $\simulator_0$ sits between the honest commiter $C$ and $R^*$, and it just forwards their messages. It follows trivially that
$M_{\Pi_{\ecom},C^*} \qc
    M_{\qecomhzero, \simulator_0}$.

    \medskip\noindent
    \textit{Hybrid $1$.}
     $\simulator_1$ interacts with $R^*$ following the protocol $\qecomhone$, which is the same as $\qecomhzero$ except that $\simulator_1$ uses the ZK simulator in Step $4$ and the decommitment phase.
From the computational zero-knowledge property, we have that
$    M_{\qecomhzero, \simulator_0} \qc
    M_{\qecomhone, \simulator_1}$.

    \medskip\noindent
    \textit{Hybrid $2$.}
     $\simulator_2$ interacts with $R^*$ following the protocol $\qecomhtwo$, which is the same as $\qecomhone$ except that $\simulator_2$  sets $c^*$ to be a commitment to $0$.
It follows from the computationally-hiding property of the commitment scheme that
$    M_{\qecomhone, \simulator_1} \qc
    M_{\qecomhtwo, \simulator_2}$.

    \medskip\noindent
    \textit{Hybrid $3$.}  $\simulator_3$ interacts with $R^*$ following the protocol $\qecomhthree$, which is the same as $\qecomhtwo$ except that $\simulator_3$ uses $\vec{\mu} = 0^{\ell}$ as the $\cdsS$ message.

    From the soundness of the ZK proof in Step $2$, we have that $c$ is not a commitment of $1$. In this case, by the security of CDS, $R^*$ does not receive $\vec{\mu}$, so the change of the message cannot be distinguished.

    Notice that Hybrid $3$ matches the description of the  simulator $\simulator$, and therefore
    $    M_{\qecomhtwo, \simulator_2} \qc
    M_{\fsocom, \simulator}$.
    and this
    finishes the proof of the first part of our lemma.

\end{proof}

\newcommand{\miniqcrypt}{\textsf{MiniQCrypt}}

\section{Multiparty (Quantum) Computation in MiniQCrypt} 
\label{sec:mpc}
Our quantum protocol realizing $\fsocom$ from quantum-secure OWF allows us to combine existing results and realize
secure computation of any two-party or multi-party classical
functionality as well as quantum circuit in MiniQCrypt.

\begin{theorem}\label{thm:mpc}
  Assuming that post-quantum secure one-way functions exist, for every
  classical two-party and multi-party functionality $\func$, there is
  a quantum protocol \cqsa-emulates $\func$.
\end{theorem}

\begin{proof}
  By Theorem~\ref{thm:ot}, we readily realize $\fot$ in
  \miniqcrypt. In the $\fot$-hybrid model, any classical functionality
  $\func$ can be realized statistically by a classical protocol in the
  universal-composable model~\cite{C:IshPraSah08}. The security can be
  lifted to the quantum universal-composable model as shown by
  Unruh~\cite{EC:Unruh10}. As a result, we also get a classical
  protocol in the $\fot$-hybrid model that \sqsa emulates
  $\func$. Plugging in the quantum protocol for $\fot$, we obtain a
  quantum protocol that \cqsa-emulates $\func$ assuming existence of
  quantum-secure one-way functions.
\end{proof}

Now that we have a protocol that realizes any classical functionality
in \miniqcrypt, we can instantiate $\func_{mpc}$ used in the work of
\cite{EC:DGJMS20} to achieve a protocol for secure multi-party quantum
computation where parties can jointly evaluate an arbitrary quantum
circuit on their private quantum input states. Specifically consider a
quantum circuit $Q$ with $k$ input registers. Let $\func_Q$ be the
ideal protocol where a trusted party receives private inputs from $k$
parties, evaluate $Q$, and then send the outputs to respective
parties. We obtain the following.

\begin{theorem}\label{thm:mpqc}
  Assuming that post-quantum secure one-way functions exist, for any
  quantum circuit $Q$, there is a quantum protocol that
  $\cqsa$-emulates the $\func_Q$.
\end{theorem} 

\bigskip

\paragraph{Acknowledgements.}
We thank the Simons Institute for the Theory of Computing for providing a meeting
place where the seeds of this work were planted. VV thanks Ran Canetti for
patiently answering his questions regarding universally composable commitments.

Most of this work was done when AG was affiliated to CWI and QuSoft. 
HL was supported by NSF grants CNS-1528178, CNS-1929901, CNS-1936825 (CAREER), CNS-2026774, a Hellman Fellowship, a JP Morgan AI Research Award, the Defense Advanced Research Projects Agency (DARPA) and Army Research Office (ARO) under Contract No.\ W911NF-15-C-0236, and a subcontract No.\ 2017-002 through Galois. FS was supported by NSF grants CCF-2041841, CCF-2042414, and CCF-2054758 (CAREER). VV was supported by DARPA under Agreement No. HR00112020023, a grant from the MIT-IBM Watson AI, a grant from Analog Devices, a Microsoft Trustworthy AI grant, and a DARPA Young Faculty Award. The views expressed are those of the authors and do not reflect the official policy or position of the Department of Defense, DARPA, the National Science Foundation, or the U.S.\ Government. 

\bibliographystyle{alpha}
\bibliography{QOT,abbrev3,crypto}

\appendix

\ifnum\llncs=1
\newpage

\section*{Supplementary materials}
\fi

\section{Preliminaries}
\label{sec:prelim}
\subsection{Basic ideal functionalities}
\label{sec:func}

In this section, we define the ideal functionalities for commitment with selective opening $\fsocom$ and (batch) oblivious transfer $\fot$ (resp. $\fpot$).

\begin{figure}[H]
  \begin{mdframed}[style=figstyle,innerleftmargin=10pt,innerrightmargin=10pt]
  \begin{center}
    \textbf{Ideal commitment with selective opening} $\func_{\socom}$
  \end{center}

\noindent  \textsf{Security parameter:} $\lambda$. \\ \textsf{Parties:} Committer $C$, Receiver $R$ and Adversary $\adv$.
\begin{description}
\item[Commit phase:] $\func_{\socom}$ receives a query
  $(\text{Commit},sid,C,R, (m_1,...,m_{r(\lambda)})$ from the committer  $C$, for some function $r$. $\func_{\socom}$ records
  $(sid,C,R,(m_1,...,m_{r(\lambda)}))$ and sends $(\text{Receipt},sid,C,R)$
  to $R$ and $\adv$. $\func_{\socom}$ ignore further
  commit messages.
\item[Decommit phase:] $\func_{\socom}$ receives a query
  $(\text{Reveal},sid,C,R,I)$, where $I$ is an index set of
  size $|I|\le r(\lambda)$.  $\func_{\socom}$ either ignores the
  message if no $(sid,C,R,(m_1,...,m_{r(\lambda)}))$ is recorded; otherwise
  $\func_{\socom}$ records $I$ and sends a
  message $(\text{Open},sid,C,R,\hat m = \{m_i: i\in I\})$ to $R$ and $\mathcal{A}$ and a message
  $(\text{Choice},sid,C,R,I)$ to $C$ and $\adv$.
\end{description}
\end{mdframed}
\begin{mdframed}[style=figstyle,innerleftmargin=10pt,innerrightmargin=10pt]
  \begin{center}
\textbf{Ideal oblivious-transfer functionality} $\fot$
  \end{center}

\noindent  \textsf{Security parameter:} $\lambda$. \\  \textsf{Parties:} Sender $S$ and receiver $R$, adversary $\adv$.
\begin{description}
\item[Sender query:] $\func_{ot}$ receives a query
  $(\text{Sender},sid,x_0,x_1)$ from $S$, where
  $x_b \in \bits^{\ell(\lambda)}$ for $b=0,1$ and $\ell(\cdot)$ is
  polynomial-bounded. $\func_{ot}$ records $(sid,x_0,x_1)$.
    \item[Receiver query:] $\func_{ot}$ receives a query
      $(\text{Receiver},sid,b)$ from $R$ and either ignores the message if
      no $(x_0,x_1)$ is recorded; otherwise $\func_{ot}$ sends
      $(\text{Reveal},sid,x_b)$ to $R$ via the control of
      $\adv$.
\end{description}
\end{mdframed}
\begin{mdframed}[style=figstyle,innerleftmargin=10pt,innerrightmargin=10pt]
  \begin{center}
\textbf{Ideal parallel oblivious-transfer functionality} $\fpot$
  \end{center}

\noindent  \textsf{Security parameter:} $\lambda$.  \\ \textsf{Parties:} Sender $S$ and receiver $R$, adversary $\adv$.
\begin{description}
\item[Sender query:] $\func_{\pot}$ receives a query
  $(\text{Sender},sid, \mathbf{x} = (x_0^i,x_1^i)_{i=1}^{\kappa(\secpar)})$ from $S$,
  where $\ell(\cdot), \kappa(\cdot)$ are polynomial-bounded and $x_b^i\in \bits^{\kappa(\secpar)}$ for
  all $i = 1,\ldots,\ell$ and $b\in \bits$. $\func_{\pot} $ records
  $(sid, \mathbf{x}=(x_0^i,x_1^i)_{i=1}^{\ell})$.
\item[Receiver query:] $\func_{\pot}$ receives a query
  $(\text{Receiver},sid, c\in\bits^\ell)$ from $R$ and either ignores
  the message if no $\mathbf{x}$ is recorded; otherwise $\func_{\pot}$
  sends $(\text{Reveal},sid,(x_{c_i}^i)_{i=1}^\ell)$ to $R$ via the control of $\adv$.
\end{description}
\end{mdframed}
\end{figure}

\subsection{Cryptographic constructions}
\label{sec:cryptoprelims}

\subsubsection{Naor's commitment scheme}\label{sec:naor}
We recall Naor's classical statistically binding commitment
   protocol~\cite{C:Naor89} from post-quantum one-way functions.
   The {\em bit commitment}
   protocol between a committer $C$ and receiver $R$, using a post-quantum
   pseudorandom
   generator $G:\{0,1\}^{\secpar} \to \{0,1\}^{3\secpar}$, proceeds as follows.

   \begin{enumerate}
     \item $R$ chooses a uniformly random string $\rho \gets \{0,1\}^{3\secpar}$ and
     sends it to $C$. (Note that the receiver $R$ of the protocol is public coin.)
     \item $C$ chooses a uniformly random string $r \gets \{0,1\}^{\secpar}$
     and sends $G(r) \oplus m\cdot \rho$ to $R$.
   \end{enumerate}

   We will sometimes succinctly describe the protocol by only referring to the committer's
   message and denote it as $\com_{\rho}(m;r)$.
   The opening of the commitment is simply
   the committer's private random coins $r$.

   \begin{theorem}[\cite{C:Naor89}]
   The protocol $(C,R)$ is statistically binding and computationally hiding against quantum
   polynomial-time adversaries assuming that
   $G$ is a post-quantum secure pseudorandom generator.
   \end{theorem}
   We remark that post-quantum pseudorandom generators
   can be constructed from post-quantum secure one-way functions~\cite{HILL99,Zhandry12}.
   Finally, we remark that the receiver can reuse $\rho$
   across many commitments sent to it.

\subsubsection{Zero knowledge protocols}

\begin{definition}
A post-quantum zero-knowledge protocol for an NP relation $\mathcal{R}$, is an interactive protocol between $P$ and $V$ that are given some input $x$ and $P$ is also given some $w$ such that $(x,w) \in \mathcal{R}$, if such $w$ exists. We require that
  \begin{description}
    \item[] \textbf{Completeness:}
      { If there exists $w$ such that  $(x,w) \in \mathcal{R}$, then
      $\Pr [ V \text{ accepts }] \geq 1 - \negl(|x|)$.}
    \item[] \textbf{Soundness:} {If for all $w$ $(x,w)\not\in \mathcal{R}$, then for all $P^*$ interacting with $V$, we have that
      $\Pr [V \text{ accepts }] \leq \negl(|x|)$.}
    \item[]  \textbf{Computational zero-knowledge:} For any $x$ such that there exists $w$ such that $(x,w) \in \mathcal{R}$ and
      any polynomial-time $V'_x$, there
      exists a polynomial-time quantum channel $\mathcal{S}_{x,V'}$ we have that \[\mathcal{P}_{V'_x} \approx_c \mathcal{S}_{s,V'}(\cdot),\]
      where $\mathcal{P}_{V'_x}$ is the quantum channel corresponding to the interaction of $V'$ with the honest prover, and both $\mathcal{P}_{V'_x}(\cdot)$ and  $\mathcal{S}_{s,V'}(\cdot)$ receive some (polynomially-large) quantum state that represents the side information of $V'_x$.
  \end{description}
\end{definition}

\begin{theorem}[\cite{Watrous09}]\label{lem:zk-watrous}
Assuming the existence of post-quantum secure one-way functions, there is a post-quantum zero-knowledge protocol for all NP relations.
\end{theorem}

\subsubsection{Yao's garbled circuits}
\begin{definition}
A garbling scheme $\mathcal{G}$ for some family of circuits $\mathcal{C}$ consists of a tuple of algorithms $(\mathsf{Garb},\mathsf{Enc},\mathsf{Eval})$ where
\begin{itemize}
    \item $\mathsf{Garb}(1^\lambda,C)$ for some $C \in \mathcal{C}$ with input length $\ell=\ell(\secpar)$
    returns a garbled circuit $\hat{C}$ and
    $2\ell$ labels $e = \big(\ell^{i}_b \in \{0,1\}^{\secpar}\big)_{i\in [\ell],b\in\bits}$. %
    \item $\mathsf{Enc}(e,x)$ outputs $\hat{x} = \big( \ell^{i}_{x_i} \big)_{i\in [\secpar]}$. %
    \item $\mathsf{Eval}(\hat{C},\hat{x})$ takes as input the garbled circuit $\hat{C}$ and the garbled input $\hat{x}$ and outputs $y$.
\end{itemize}
We require the following properties of $\mathcal{G}$:
\begin{description}
\item[Correctness:] $\Pr[\mathsf{Garb}(1^\lambda,C) \rightarrow (\hat{C},e), \mathsf{Eval}(\hat{C},\mathsf{Enc}(e,x)) = f(x)] = 1$
\item[Security:] There is a polynomial-time simulator $\mathsf{GarbSim}$ such that for all circuits $C \in \mathcal{C}$ and all input $x \in \{0,1\}^{\ell}$, the following two distributions are computationally indistinguishable:
\[ \big( (\hat{C},\hat{e}): (\hat{C},e) \gets \mathsf{Garb}(1^\lambda,C); \hat{e} \gets \mathsf{Enc}(e,x) \big) \stackrel{c}{\approx}  \big( (\hat{C},\hat{e}) \gets \mathsf{GarbSim}(1^\secpar, 1^{|C|}, C(x)) \big) \]
\end{description}
\end{definition}

\begin{lemma}[\cite{Yao86}]
  Assuming the existence of post-quantum secure one-way functions, there is a post-quantum secure garbling scheme for the family of all polynomial-size circuits.
\end{lemma}
\subsection{Quantum information basics}
\label{sec:quantum}
We review now the concepts and notation of Quantum Computation that are used in
this work. We refer to Ref. \cite{NC02} for a detailed introduction of
these topics.

A pure quantum state of $n$ qubits is a unit vector in the Hilbert space
$\left\{\complex^{2}\right)^{\otimes n}$, where $\otimes$ is the Kroeneker (or tensor) product.
The basis for such Hilbert space is $\{\ket{i}\}_{i \in \{0,1\}^n}$.
For some quantum state $\ket{\psi}$, we denote $\bra{\psi}$ as its conjugate transpose. The inner product between two vectors $\ket{\psi}$ and $\ket{\phi}$ is denoted by
$\braket{\psi}{\phi}$ and their outer product as
$\ketbra{\psi}{\phi}$. A mixed state is a  (classical) probabilistic distribution of pure quantum states. The mixed state corresponding to having the quantum state $\ket{\psi_i}$ with probability $p_i$ (with $\sum_i p_i  = 1$) is represented by its density matrix $\rho = \sum_i p_i \kb{\psi_i}$.

Specifically for qubits, there are two important basis that we consider in this work. The computational (or $+$) basis consists of $\{\ket{0}, \ket{1}\}$ and the Hadamard (or $\times$) basis consists of
$\{\ket{+} = 1/\sqrt{2}(\ket{0}+\ket{1}), \ket{-} = 1/\sqrt{2}(\ket{0}-\ket{1})\}$. For $b \in \{0,1\}$ and $\theta \in \{+,\times\}$, we define $\ket{b}_{\theta} = \begin{cases} \ket{b}, & \text{if } \theta = + \\
1/\sqrt{2}(\ket{0} + (-1)^b\ket{1}), & \text{if } \theta = \times
\end{cases}.$

We describe now the operations that can be performed on quantum states. If we measure a quantum state $\rho$ with some projective measurement $\Pi_0$ $I - \Pi_1$, we have that the output is $b$ with probability $Tr(\Pi_0\rho)$.

\subsection{Quantum machine model and quantum indistinguishability}

We review the quantum machine model and notions of
indistinguishability as described in~\cite{HSS15}. A {\em
  quantum interactive machine} (QIM) $\mac$ is an ensemble of
interactive circuits $\{\mac_x\}_{x \in I}$. The index set $I$ is
typically the natural numbers $\mathbb{N}$ or a set of strings
$I\subseteq \{0,1\}^*$.  For each value $\secpar$ of the security
parameter, $\mac_\secpar$ consists of a sequence of circuits
$\{\mac_\secpar^{(i)}\}_{i=1,...,\ell(\secpar)}$, where
$\mac_\secpar^{(i)}$ defines the operation of $\mac$ in one round $i$
and $\ell(\secpar)$ is the number of rounds for which $\mac_\secpar$
operates (we assume for simplicity that $\ell(\secpar)$ depends only
on $\secpar$). We omit the scripts when they are clear from the
context or are not essential for the discussion. Machine
$\mac_\secpar$ (i.e., each of the its constituent circuits) operates
on three registers: a state register $\qreg{S}$ used for input and
workspace; an output register $\qreg{O}$; and a network register
$\qreg{N}$ for communicating with other machines. The size (or running
time) $t(\secpar)$ of $\mac_\secpar$ is the sum of the sizes of the
circuits $\mac_\secpar^{(i)}$. We say a machine is polynomial time if
$t(\secpar)=\poly(\secpar)$ and there is a deterministic classical
Turing machine that computes the description of $\mac_\secpar^{(i)}$
in polynomial time on input $(1^\secpar,1^i)$.

When two QIMs $\mac$ and $\mac'$ interact, they share network register
$\qreg{N}$. The circuits $\mac_\secpar^{(i)}$ and ${\mac'}_\secpar^{(i)}$
are executed alternately for $i=1,2,...,\ell(\secpar)$. When three or
more machines interact, the machines may share different parts of
their network registers (for example, a private channel consists of a
register shared between only two machines; a broadcast channel is a
register shared by all machines). The order in which machines are
activated may be either specified in advance (as in a synchronous
network) or adversarially controlled.

A non-interactive quantum machine (referred to as QTM hereafter) is a
QIM $\secpar$ with network register empty and it runs for only one
round (for all $\secpar$). This is equivalent to the {\em quantum
  Turing machine} model (see~\cite{Yao93}). A classical interactive
Turing machine (ITM) is a special case of a QIM, where the registers
only store classical strings and all circuits are classical
(Cf.~\cite{Can00,Can01}).

\paragraph{Quantum indistinguishability.} Let
$\rho = \{\rho_\secpar\}_{\secpar \in \mathbb{N}}$ and
$\eta = \{\eta_\secpar\}_{\secpar \in \mathbb{N}}$ be ensembles of
mixed states indexed by $\secpar \in\mathbb{N}$, where $\rho_\secpar$
and $\eta_\secpar$ are both $r(\secpar)$-qubit states for some
polynomial-bounded function $r$. 

We state the indistinguishability of quantum states proposed by
Watrous~\cite[Definition 2]{Watrous09}.

\begin{definition}[$(t, \veps)$-indistinguishable states]
\label{def:qc}
We say two quantum state ensembles
$\rho = \{\rho_\secpar\}_{\secpar \in \mathbb{N}}$ and
$\eta = \{\eta_\secpar\}_{\secpar \in \mathbb{N}}$ are
$(t,\veps)$-indistinguishable, denoted $\rho \qc^{t, \veps} \eta$, if
for every $t(\secpar)$-time QTM $\envt$ and any mixed state
$\sigma_\secpar$,

\[ \left|\Pr[\envt(\rho_\secpar \otimes \sigma_\secpar) =1] - \Pr[
    \envt(\eta_\secpar \otimes \sigma_\secpar) = 1]\right| \leq
  \veps(\secpar) \, . \]

\end{definition}

The states $\rho$ and $\eta$ are called {\em quantum computationally
  indistinguishable}, denoted $\rho \qc \eta$, if for every polynomial
$t(\secpar)$, there exists a negligible $\veps(\secpar)$ such that
$\rho_\secpar$ and $\eta_\secpar$ are
$(t,\veps)$-indistinguishable. The definition subsumes classical
distributions as a special case, which can be represented by density
matrices that are diagonal in the standard basis.

Then we recall indistinguishability of QTMs~\cite{HSS15}, which is equivalent to quantum computationally indistinguishable
super-operators proposed by Watrous~\cite[Definition 6]{Watrous09}.
\begin{definition}[$(t, \veps)$-indistinguishable QTMs]
  \label{def:qcm}
  
  We say two QTMs $\mac_1$ and $\mac_2$ are
  \emph{$(t,\veps)$-indistinguishable}, denoted
  $\mac_1 \approx_{qc}^{t, \veps} \mac_2$, if for any
  $t(\secpar)$-time QTM $\envt$ and any mixed state
  $\sigma_\secpar \in \density{\rspace{S} \otimes \rspace{R}}$, where
  $\rspace{R}$ is an arbitrary reference system,

\[ \left| \Pr[ \envt((\mac_1\otimes
    \identity_{\lin{\rspace{R}}})\sigma_\secpar) = 1 ] - \Pr [
    \envt((\mac_2\otimes \identity_{\lin{\rspace{R}}})\sigma_\secpar)
    = 1] \right| \leq \veps(\secpar) \, .\] 
Machines $\mac_1$ and $\mac_2$ are called {\em quantum computationally
  indistinguishable}, denoted $\mac_1 \qc \mac_2$, if for every
polynomial $t(\secpar)$, there exists a negligible $\veps(\secpar)$
such that $\mac_1$ and $\mac_2$ are $(t,\veps)$-computationally
indistinguishable.
\end{definition}

\section{Proof of \Cref{thm:ot}}
\label{ap:qot}

We split our proof into two steps. In \Cref{sec:comalice}, we prove security against adversaries corrupting the sender, and in \Cref{sec:combob} we prove security against adversaries corrupting the receiver. As we previously mentioned, our proofs consist of a straightforward adaptation of previous results~\cite{C:BBCS91,DFL+09,EC:Unruh10} to our setting/language.

\subsection{Security against a Malicious Sender}
\label{sec:comalice}

\begin{figure}[!ht]
  \begin{mdframed}[style=figstyle,innerleftmargin=10pt,innerrightmargin=10pt]
{\sf Inputs:} Environment generates inputs: chosen bit $c$ is given to honest (dummy) $R$; and input to $\adv$ is passed through $\simulator$.

\begin{enumerate}
\item (Initialization) $\simulator$ behaves as an honest $R$ does in the real protocol.
\item (Checking) Simulate $\fsocom$ and commitment to a dummy
  message. Once receiving the checking set $T$, measure them and run
  $\adv$ on the outcome. $\simulator$ aborts if $\adv$ aborts.

\item (Partition Index Set) Let $\hat \theta^A$ be the basis received
  from $\adv$. $\simulator$ measures the remaining qubits under
  $\theta^A$, and obtains $\hat x^B$. $\simulator$ then randomly
  partitions the indices into $I_0$ and $I_1$ and sends them to $\adv$.

\item (Secret Transferring) Once receiving $(f, m_0, m_1)$,
  $\simulator$ computes $s_0' := m_0 \oplus f({\hat x^B}|_{I_0})$ and
  $s_1' := m_1 \oplus f({\hat x^B}|_{I_1})$. $\simulator$ gives $\fot$ the pair $(s_0', s_1')$. Outputs whatever $\adv$ outputs in the end.
\end{enumerate}
\end{mdframed}
\caption{Simulator for malicious sender in $\qotbbcs$.}
\label{fig:sim-alice-qot}
\end{figure}

\begin{lemma}
$\qotbbcs$ $\cqsa$-realizes $\fot$ in the $\fsocom$ hybrid model against a malicious sender.
\end{lemma}
\begin{proof}
Given an adversary $\adv$ that corrupts $S$, we construct a simulator $\simulator$ as described in \Cref{fig:sim-alice-qot}. Our goal is to show that
\begin{align}\label{eq:sim-alice-ot}
    \mac_{\qotbbcs,\adv} \qc \mac_{\fot,\simulator}.
\end{align}
In order to prove \Cref{eq:sim-alice-ot}, we provide four hybrids.

  \medskip \noindent
  \textit{Hybrid $0$.} $\simulator_0$  interacts with $\adv$ following the protocol $\qothzero$, where $\simulator_0$ simulates honest $R$ and $\fsocom$ in $\qotbbcs$. It follows trivially that \begin{align}\label{eq:receiver-security-h0}
M_{\qotbbcs,\adv} \approx_{qc} M_{\qothzero,\simulator_0}.
\end{align}

  \medskip \noindent
  \textit{Hybrid $1$.} $\simulator_1$ interacts with $\adv$ following the protocol $\qothone$, which is the same as $\qothzero$ the following differences:
\begin{enumerate}
    \item  $\simulator_1$ waits the subset $T$, measures the corresponding qubits on the basis $\tilde \theta^B$ with outcome $\tilde x^B$ and then simulates the opening of $\fot$ with $\tilde \theta^B$ and $\tilde x^B$.
   \item After receiving $\hat \theta^A$, $\simulator_1$ measures the remaining qubits on basis $\hat \theta^B$ and continue as $\qothzero$.
\end{enumerate}

 Since the only difference here is that $R$ delays the measurements, but since such operations commute, we have that
\begin{align}\label{eq:receiver-security-h1}
M_{\qothzero,\simulator_0} \approx_{qc} M_{\qothone,\simulator_1}.
\end{align}

  \medskip \noindent
  \textit{Hybrid $2$.}
$\simulator_2$ interacts with $\adv$ following the protocol $\qothtwo$, which is the same as $\qothone$
 except that $\simulator_2$ measures the remaining qubits (not in $T$) on basis $\hat \theta^A$ and continue as $\qothtwo$.

The only difference here is that the basis used by $\simulator_2$ to measure the qubits not in $T$ are not the same.  Since these values are never revealed to $\adv$, we have that
\begin{align}\label{eq:receiver-security-h2}
M_{\qothone,\simulator_1} \approx_{qc} M_{\qothtwo,\simulator_2}.
\end{align}

  \medskip \noindent
  \textit{Hybrid $3$.}
  Interaction of $\simulator$ from \Cref{fig:sim-alice-qot} with $\tilde{\Pi}_{\fot}$.

Notice that $\simulator$ performs the exact same operations as $\simulator_2$ in $\qothtwo$, and the only difference is that $\simulator$ uses the information (that the $\simulator_2$ already had) to extract the two bits and inputs them into $\fot$. Therefore
\begin{align}\label{eq:receiver-security-simulation}
  M_{\qothtwo,\simulator_2}.
  \approx_{qc}
  M_{\fot,\simulator}
\end{align}

\Cref{eq:sim-alice-ot} follows from
\Cref{eq:receiver-security-h0,eq:receiver-security-h1,eq:receiver-security-h2,eq:receiver-security-simulation}.
\end{proof}

\subsection{Security against malicious receiver}
\label{sec:combob}

As a standard proof technique, we first describe a variant
$\qotepr$ based on EPR-pairs. Then a sampling
framework by Bouman and Fehr~\cite{C:BouFeh10} ensures a min-entropy
bound that is the core of establishing simulation.

  \begin{mdframed}[style=figstyle,innerleftmargin=10pt,innerrightmargin=10pt]
  \begin{center}
    Equivalent protocol  $\qotepr$
  \end{center}

{\sf Inputs:} $S$ gets input two $\ell$-bit strings $s_0$ and $s_1$, $R$ gets a bit $c$.

\begin{enumerate}
\item (Initialization) $S$ generates $n$ pairs of EPR
  $\ket{\Psi}^{\otimes n} = [\frac{1}{\sqrt{2}}(\ket{00} +
  \ket{11})]^{\otimes n}$, and sends $R$ $n$ halves of these EPR
  pairs. $S$ chooses $\tilde \theta^A \in \{+, \times\}^n$ at
  random, but doesn't measure her shares of the EPR pairs.
\item (Checking) As in $\qotbbcs$, $S$ picks the random subset $T$
  and obtains $\{\tilde \theta_i^B, \tilde x_i^B\}$ from $\fsocom$. Then
  for each $i\in [T]$ where $S$ measures her $i$th share of EPR
  under bases $\tilde \theta_i^A$. If the outcome
  $\tilde x^A_i \neq \tilde x^B_i$ but
  $\tilde \theta_i^A = \tilde \theta_i^B$, abort. Once the checking is
  done, $S$ continues and measures her remaining qubits under
  $\hat \theta^A$ to obtain $\hat x^A$.

\item (Partition Index Set) Same as $\qotbbcs$.
\item (Secret Transferring) Same as $\qotbbcs$.
\end{enumerate}
\end{mdframed}

Notice that this protocol is equivalent to
$\qotbbcs$ from receiver's perspective, and therefore we have the following implication.

\begin{lemma}\label{lem:epr-to-bbcs}
$\qotepr$ $\cqsa$-implements $\qotbbcs$ in the $\fsocom$ hybrid model against corrupted receiver.
\end{lemma}

The high-level
approach is:

\begin{itemize}
\item Interpret the checking phase in $\qotepr$ as a sampling game
  over qubits.
    \item Analysis of the sampling game shows that if $R$ passes the
      checking phase, then the {\em real} joint state of $S$ and $R$
      after the checking phase in the protocol will be negligibly
      close to an {\em ideal} state.
    \item Finally we argue that if one measured systems by $S$ in the
      {\em ideal} state and gets a string $x$, then no matter how $R$
      partitions the index sets $(I_0, I_1)$, there exists a $c$ such
      that a large amount of min-entropy is preserved in
      $x|_{I_{1-c}}$. In addition, this bit $c$ can be derived
      efficiently.
\end{itemize}

Thus we see that if $R$ indeed passes the checking phase in the real
protocol, $f(\hat x^A|_{I_{1-c}})$ will be statistically close to
uniform, except with negligible probability. This enables us to
construct a simulator when the receiver is corrupted below, and we can prove
the following.
\begin{proposition}\label{lem:security-bob-epr}
$\qotepr$ $\cqsa$-realizes $\fot$ in the $\fsocom$ hybrid model against corrupted receiver.
\end{proposition}

  \begin{mdframed}[style=figstyle,innerleftmargin=10pt,innerrightmargin=10pt]
  \begin{center}
  Simulator   $\simulator$ (when receiver is corrupted)
  \end{center}
{\sf Inputs:} Environment generates inputs: $s_0$ and $s_1$ are given
to honest (dummy) $S$; and input to $\adv$ is passed through $\simulator$.

\begin{enumerate}
\item (Initialization) $\simulator$ initializes an execution with corrupted $R$, just as in $\qotbbcs$.
\item (Checking) $\simulator$ does the checking procedure as in
  $\qotepr$, and the commitment is simulated internally. Therefore
  $\simulator$ sees all $(\tilde x^B_i, \tilde \theta^B_i)$ that
  corrupted $R$ sent to $\fsocom$.

\item (Partition Index Set) $\simulator$ expects to receive $(I_0,I_1)$ from $\adv$.

\item (Secret Transferring) $S$ sends $(s_0,s_1)$ to the ideal
  functionality $\fot$. $\simulator$ sets $c\in \bits$ to be such that
  $wt(\res{\hat \theta^A}{I_c} \oplus \res{\hat \theta^B}{I_c}) \leq
  wt(\res{\hat \theta^A}{I_{1-c}} \oplus \res{\hat
    \theta^B}{I_{1-c}})$. (That is, the Hamming distance between $\hat
  \theta^A$ and $\hat \theta^B$, restricted to $I_{1-c}$, is larger.)
  Send $c$ to the (external) $\fot$ and obtain $s_c$. $\cS$ then sends
  $f\in_R \mathbf{F}$, $m_c := s_c \oplus f(\res{x^A}{I_c})$ and $m_{1-c}
  \in_R \bits \ell$ to $\adv$. Output whatever $\adv$ outputs in the
  end.
\end{enumerate}

\end{mdframed}

\begin{proof}
  Observe that the simulation of $\simulator$ differs from the
  real-world execution only in the last \emph{secret transferring}
  phase: in both cases $m_c = s_c \oplus f(\res{\hat x^A}{I_c})$, but
  $m_{1-c} = s_{1-c} \oplus f(\res{\hat x^A}{I_{1-c}})$ in $\qotepr$,
  while during simulation $\simulator$ sets $m_{1-c}\in_R
  \bits^\ell$. However, as we argue below in Theorem~\ref{thm:qot-bob}, after the checking phase,
   system $A$ of $S$ restricted to $I_{1-c}$ has high min-entropy even
  conditioned on the adversary's view. Hence $f$ will effectively
  extract $\ell$ uniformly random bits.
\end{proof}

\begin{corollary}
$\qotbbcs$ $\cqsa$-realizes $\fot$ in the $\fsocom$ hybrid model against corrupted receiver.
\end{corollary}
\begin{proof}
Follows directly from \Cref{lem:epr-to-bbcs,lem:security-bob-epr}.
\end{proof}

\begin{theorem}
  Let $M_0$ and $M_1$ be the two message systems generated by $S$ in
  $\qotepr$. Then, there exists $c \in \bits$ such that $M_{1-c}$ is
  close to uniformly random and independent of the receiver's view. Namely, let
  $\rho_{M_{1-c}M_c B}$ be the final state of $R$,
  there exists a state $\sigma_{M_c B}$ such that
\[
\td(\rho_{M_{1-c}M_c B}, \frac{1}{2^{\ell}}\mathbb{I}\otimes \sigma_{M_{c}B}) \leq  negl(\lambda) \ . \]
\end{theorem}
\label{thm:qot-bob}
\begin{proof}
This theorem follows from~\cite[Theorem 4]{C:BouFeh10} by observing
  that bit $c$ determined by the simulator $S$ is the correct one. For completeness, we present here some details.

  Consider the joint state $\ket{\phi_{AE}}$ right before the checking
  phase of $\qotepr$, consisting of the $n$ EPR pairs plus potentially
  some additional quantum system on the receiver's side. This checking
  corresponds exactly to a sampling game on quantum states analyzed
  in~\cite{C:BouFeh10}. It is shown that for any constant $\delta > 0$, the
  real state, after $S$ has measured the selected qubits, is
  $\veps_{samp}$-close to an ideal state with
  $\veps_{samp} \le \sqrt 6 \exp(- \alpha n \delta^2/{100})$. The ideal state
  is a superposition over basis vectors $\ket{x}_{\hat{\theta}^B}$ of
  relative hamming weight at most $\delta$ with respect to the basis
  $\hat{\theta}^B$.  This remains the case after $R$ announces the
  sets $I_0$ and $I_1$, and we can also view the qubits $A_{I_c}$ as
  part of the adversarial system $E$, where $c \in \bits$ is such that
  $wt(\res{\hat \theta^A}{I_{c}} \oplus \res{\hat \theta^B}{I_{c}})
  \leq wt(\res{\hat \theta^A}{I_{1-c}} \oplus \res{\hat
    \theta^B}{I_{1-c}})$. Note that (by Hoeffding's inequality) except
  with probability $\veps_{hof} \le \exp(-2\eta^2(1-\alpha)n)$, the number
  of positions $i \in I_{1-c}$ with
  $\hat \theta^A_i \neq \hat \theta^B_i$ is at least
  $( \frac1 2 - \eta)(1-\alpha)n$.

  It follows from Fact~\ref{fact:SmallSuperpos} below that (for the
  ideal state)
\[
\HH_{min}(\hat{X}_{1-c} |A_{I_c} E) \geq 
(\frac12 - \delta)(1-\alpha)n - h(\delta)n \, ,
\]
except with negligible probability, where
$\hat{X}_{1-c} = \res{\hat{X}^A}{I_{1-c}}$ and the left hand side
should be understood as conditioned on all the common classical
information, $\hat \theta^A, \hat \theta^B$ etc. By basic properties
of the min-entropy, the same bound also applies to
$\HH_{min}(\hat{X}_{1-c} |\hat{X}_{c} E)$. It then follows from
privacy amplification~\cite{TCC:RenKon05} that if
$( \frac1 2 - \eta) (1-\alpha)n - h(\delta)n - \ell \ge \gamma n$, then the
extracted string $S_{1-c}$ is $\veps_{pa} \le 2^{-\frac 1 2 \gamma n}$ to uniform given $X_c$ (and
hence also given $S_c$), the quantum system $E$, and all common
classical information. Collecting all the ``errors'' encountered on
the way, the distance to uniform becomes 
\[ \veps: = \veps_{samp}+\veps_{hof}+\veps_{pa} \le \sqrt 6 \exp(-
\alpha n \delta^2/100) + \exp(-2\eta^2(1-\alpha)n) + 2^{-\frac 1 2
  \gamma n}=  \negl(\secpar) \, , \]
by setting $\alpha = , \eta = , \gamma= , \delta = $.  
\end{proof}

\begin{fact}[{\cite[Corollary 1]{C:BouFeh10}}]\label{fact:SmallSuperpos}
Let $\ket{\phi_{AE}}$ be a superposition on states of the form $\ket{\bd{x}}_{\theta'}\ket{\phi_E}$ with $|w(\bd{x})| \leq \delta$ and $\delta < 1/2$,
and let the random variable $\bd{X}$ be the outcome of measuring $A$ in basis $\theta \in \{+,\times\}^n$. Then
$$
\HH_{min}(\bd{X} |E) \geq wt(\theta \oplus \theta') - h(\delta)n.
$$
where $h(p): = -p\log p - (1-p)\log(1-p)$ is the Shannon binary entropy.
\label{fact:minentropy}
\end{fact}
\begin{fact} [Privacy Amplification {\cite[Theorem 1]{TCC:RenKon05}}] Let $\rho_{XE}$ be a hybrid state with classical $X$ with the form $\rho_{XE} = \sum_{x\in \mathcal{X}} P_X \bk{x}\otimes \rho_E^x$. Let $\mathbf{F}$ be a family of universal hash functions with range $\bits^\ell$, and $F$ be chosen randomly from $\mathbf{F}$. Then $K = F(X)$ satisfies
\[ D(\rho_{KFE}, \frac{1}{2^\ell}\mathbb{I}_{K}\otimes \rho_{FE}) \leq
  \frac{1}{2}\cdot 2^{-\frac{1}{2}(\mathbf{H}_{min}(X|E) - \ell)} \,
  . \]
\label{fact:pa}
\end{fact}

\section{Proof of \Cref{thm:parallel}}
\label{ap:parallel}
\begin{proof}
  Let $\pi$ be the protocol realizing $\func$ between two
  parties. %
  In each round, the parties $A$ and $B$ either compute a message on
  the network register or make a call to the functionality $\funcg$.
  The parallel repetition of $\pi$, assuming it is repeated $\ell$
  times, is denoted by $\pi^{||}$ and works as
  expected. %
  In each round, $A^{\|}$ and $B^{\|}$ either run $A$ and $B$ to
  compute the messages in each execution of $\pi$, or make a call to
  the ideal functionality $\funcg^{||}$.  %

  The simulator
  $\simulator^{||} = (\simulator_A^{||},\simulator_B^{||})$ for
  $\pi^{||}$ works as follows. We will describe $\simulator_A^{||}$
  for an adversary $\adv$ corrupting $A^{\|}$ for simplicity;
  $\simulator_B^{||}$ works in a completely analogous way.
  In each round $\simulator_A^{||}$ receives a tuple of registers from
  $\adv$ and applies the simulator $\simulator$ for $\pi$ on each of
  them. Also the calls to $\funcg^{\|}$ will be simulated by
  $\simulator$ in each execution. When $\simulator$ sends a message to
  $\func$, $\simulator_A^{\|}$ collects all of them, and forward to
  $\func^{\|}$. Likewise, any message $\func^{\|}$ returns,
  $\simulator_A^{\|}$ will split them and forward to each
  $\simulator$. Because of straight-line simulation, the actions of
  $\simulator_A^{\|}$ are well-defined   and moreover $\simulator^{\|}$ is efficient if $\simulator$ is efficient (with an unavoidable multiplicative factor of $\ell$).

  We show that
  $\mac_{\pi^{\|},\adv} \qc
  \mac_{\func^{\|},\simulator_A^{\|}}$\footnote{If $\pi$ \sqsa{}
    emulates $\func$, the indistinguishability here becomes
    statistical too.}. This is done by a simple hybrid argument. Let
  $\mac_0 = \mac_{\pi^{\|},\adv}$, and for $i = 1,\ldots, \ell$, let
  $\mac_i$ be as $\mac_{i-1}$ except that the $i$th execution of $\pi$
  with $\adv$ is simulated by $\simulator$. We can see that
  $\mac_{i-1} \qc \mac_i$ holds for all $i \in [\ell]$. This is
  because one can think of $\adv$ and all executions other than the
  $i$th as another adversary $\adv'$ attacking $\pi$. Note that the
  first $i-1$ executions involve simulator $\simulator$, and it is
  crucial that $\simulator$ is straight-line to ensure $\adv'$ is
  well-defined.
  Then $\mac_{i-1} \qc \mac_i$ follows by the security
  of $\pi$. Note that $\mac_{\ell}$ is exactly
  $\mac_{\func^{\|},\simulator_A^{\|}}$, and hence we conclude that
  $\mac_{\pi^{\|},\adv} \qc \mac_{\func^{\|},\simulator_A^{\|}}$.
\end{proof}

\section{Proof of \Cref{lem:socom-sender}}
\label{ap:proof-sender-socom}
In the original protocol, after the first message from the $C$ to $R$, the joint state is

\[
\frac{1}{|\mathcal{R}|}\sum_{\rho}
\sum_{\hat{c}_1,...,\hat{c}_k} p_{\hat{c}_1,...,\hat{c}_k} \kb{\hat{c}_1,...,\hat{c}_k} \otimes \rho_{\hat{c}_1,...,\hat{c}_k},
\]
for some quantum states $\rho_{\hat{c}_1,...,\hat{c}_k}$. Notice that $\hat{c}_i$'s depend on $\rho$, but we leave such a dependence implicit.

Let $\mathcal{P}$ be the quantum channel corresponding to the honest behaviour of the receiver and the malicious behaviour of $\adv$, and $\mathcal{S}$ is the output of the simulator after interacting with $\adv$. Our goal is to show that
\begin{align}\label{eq:difference-simulation-commit-overall}
\frac{1}{|\mathcal{R}|}\sum_{\rho}
\sum_{\hat{c}_1,...,\hat{c}_k} p_{\hat{c}_1,...,\hat{c}_k}
\kb{\hat{c}_1,...,\hat{c}_k} \otimes  \left(\mathcal{P}(\rho_{\hat{c}_1,...,\hat{c}_k}) - \mathcal{S'}(\rho_{\hat{c}_1,...,\hat{c}_k})\right)  \leq \negl(\lambda),
\end{align}

Let $\mathcal{B}$ be the set of $r$ such that there exists some $\tilde{c} = \com_{\rho}(m,r) = \com_{\rho}(m',r')$, for $m \ne m'$. We have from Naor's commitment scheme that $|\mathcal{B}| \leq \negl(\lambda) |\mathcal{R}|$, so we now focus on proving
\begin{align}\label{eq:difference-simulation-commit-good-r}
\frac{1}{|\mathcal{R}|}\sum_{\rho \not\in \mathcal{B}}
\sum_{\hat{c}_1,...,\hat{c}_k} p_{\hat{c}_1,...,\hat{c}_k}
\kb{\hat{c}_1,...,\hat{c}_k} \otimes  \left(\mathcal{P}(\rho_{\hat{c}_1,...,\hat{c}_k}) - \mathcal{S'}(\rho_{\hat{c}_1,...,\hat{c}_k})\right)  \leq \negl(\lambda),
\end{align}
which will imply \Cref{eq:difference-simulation-commit-overall}.

Since we fix $\rho$ and the first register is held by the Receiver, we can assume that the latter is measured and we proceed with the analysis for each of the values individually: we show that for each $r \not\in \mathcal{B}$ and $\hat{c}_1,...,\hat{c}_k$, we have that
\begin{align}\label{eq:difference-simulation-commit}
\kb{\hat{c}_1,...,\hat{c}_k} \otimes  \left(\mathcal{P}(\rho_{\hat{c}_1,...,\hat{c}_k}) - \mathcal{S'}(\rho_{\hat{c}_1,...,\hat{c}_k})\right) \leq  \negl(\lambda),
\end{align}
and by convexity we have \Cref{eq:difference-simulation-commit-good-r}.

We split our argument in two cases.

\paragraph{Invalid commitment.} Let us assume that there exists some $\hat{c}_i$ that is not a valid commitment. In this case, by the soundness of ZK protocol, $R$ aborts with probability $1 - \eps$, for some $\eps = \negl(\lambda)$. In this case, the final state would be
\[\kb{\hat{c}_1,...,\hat{c}_k} \otimes \left((1-\eps)\kb{\bot} + \eps \gamma_{\hat{c}_1,...,\hat{c}_k}\right),\]
for some quantum state $\gamma_{\hat{c}_1,...,\hat{c}_k}$.

On the other hand, $\simulator$ always aborts since she is not able to extract all messages and the final state is also
\[\kb{\hat{c}_1,...,\hat{c}_k} \otimes \kb{\bot},\]
which implies that
\Cref{eq:difference-simulation-commit} holds in this case.

\paragraph{All commitments are valid.}
We now assume that all commitments are valid, i.e., for each $i \in [m]$, there exists exactly a single string $\tilde{m}_i$ such that $\hat{c}_i$ is a commitment of $\tilde{m}_i$. 
We consider two subcases here. First, for all $i \in I$, $\tilde{m}_i = \hat m_i$.
Notice that in this case, when interacting with $\adv$, $\simulator$ has the exact same behaviour as the honest $R$ in the original protocol. Therefore, for such $\rho$ and $\hat{c}_1,...,\hat{c}_k$ we have that
\[\mathcal{P}(\rho_{\hat{c}_1,...,\hat{c}_k}) = \mathcal{S'}(\rho_{\hat{c}_1,...,\hat{c}_k}),\]
and therefore
\Cref{eq:difference-simulation-commit} holds trivially.

In the second subcase, there exists some  $i \in I$ s.t. $\tilde{m}_i \ne \hat m_i$. By the soundness of the ZK protocol, the final state would be
\[\kb{\hat{c}_1,...,\hat{c}_k} \otimes \left((1-\eps)\kb{\bot} + \eps \gamma'_{\hat{c}_1,...,\hat{c}_k}\right),\]
for some quantum state $\gamma' _{\hat{c}_1,...,\hat{c}_k}$, whereas the final state of the simulation is
\[\kb{\hat{c}_1,...,\hat{c}_k} \otimes \kb{\bot},\]
and therefore
\Cref{eq:difference-simulation-commit} follows.

\section{Proof of \Cref{lem:socom-receiver}}
\label{ap:proof-receiver-socom}
  \medskip \noindent
  \textit{Hybrid $0$.}
We first consider a protocol $\qsocomhzero$ where $\simulator_0$ sits in between $\adv$ and honest $C$ in the protocol, just passing the information back and forth between the parties. It follows trivially that
$    M_{\Pi_{\socom},\adv} \approx_{qc} M_{\Pi_{\qsocomhzero},\simulator_0}$.

  \medskip \noindent
  \textit{Hybrid $1$.}
Now, the protocol $\simulator_1$ proceeds in $\qsocomhone$, which is the same as $\qsocomhzero$ except that in the decommitment phase, $\simulator_1$ runs the ZK simulator instead of running the ZK protocol.

By the computational zero-knowledge property of the protocol, we have that $M_{\Pi_{\qsocomhzero},\simulator_0} \approx_{qc} M_{\Pi_{\qsocomhone},\simulator_1}$.

  \medskip \noindent
  \textit{Hybrid $2$.}
Finally, we consider our final simulator $\simulator$, that interacts with $\adv$ and  $\fsocom$. Notice in this hybrid, $\adv$ receives commitments of $0$  instead of the real messages $m_i$ and then $\adv$ answers with a set $I$. By the computationally hiding property of Naor's commitment scheme, $\adv$ cannot distinguish between the commitment of $0$ and commitment of $m_i$ in the first message, so the distribution of $I$ is the same, up to negligible factors.

$\simulator$ then inputs $I$ to the ideal functionality to be able to retrieve the messages $(m_i)_{i \in I}$, and runs the ZK simulator to prove that the commitments $(c_i)_{i \in I}$ are commitments of $(m_i)_{i \in I}$ and that the other commitments are valid. The two ZK simulation, the first one with commitments to $(m_i)$ and the second one with commitments to $0$, cannot be distinguished, otherwise the ZK simulator could be used to break the computational hiding property of Naor's commitment scheme.

It follows that
$    M_{\Pi_{\qsocomhone},\simulator_1}
        \approx_{qc}
        M_{\fsocom,\simulator}$, which finishes the proof. 

\end{document}